\documentclass{article}
\usepackage{amssymb}
\usepackage{amsmath}
\usepackage{amsfonts}
\usepackage{version}
\usepackage{appendix}
\usepackage{portland}
\usepackage{lscape}
\usepackage{setspace}

\newtheorem{theorem}{Theorem}

\newtheorem{corollary}{Corollary}

\newtheorem{proposition}{Proposition}
\newtheorem{remark}{Remark}

\newenvironment{proof}[1][Proof]{\noindent \textbf{#1.} }{\  \rule{0.5em}{0.5em}}
\input{tcilatex}
\begin{document}

\title{The fundamental properties of time varying AR models with non
stochastic coefficients}
\author{M. Karanasos$^{\dagger }$, A. G. Paraskevopoulos and S. Dafnos \\
\textit{Brunel University, London, UK}}
\date{This draft: March 8th 2014\\
}
\maketitle

\begin{abstract}
The paper examines the problem of representing the dynamics of low order
autoregressive (AR) models with time varying (TV) coefficients. \ The
existing literature computes the forecasts of the series from a recursion
relation. \ Instead, we provide the linearly independent solutions to TV-AR
models. \ Our solution formulas enable us\ to derive the fundamental
properties of these processes, and obtain explicit expressions for the
optimal predictors. \ We illustrate our methodology and results with a few
classic examples amenable to time varying treatment, e.g, periodic,
cyclical, and AR models subject to multiple structural breaks.

\textbf{Keywords}: abrupt breaks, covariance structure, cyclical processes,
homogeneous and particular solutions, optimal predictors, periodic AR models.

\vspace{1.5in}

{\footnotesize We gratefully} {\footnotesize acknowledge the helpful
conversations we had with L. Giraitis, G. Kapetanios and A. Magdalinos in
the preparation of the paper. We would also like to thank R. Baillie, L.
Bauwens, M. Brennan, D. van Dijk, W. Distaso, C. Francq, C. Gourieroux, E.
Guerre, M. Guidolin, A. Harvey, C. Hommes, S. Leybourne, P. Minford, A.
Monfort, C. Robotti, W. Semmler, R. Smith, T. Ter\"{a}svirta, P. Zaffaroni,
and J-M Zakoian for suggestions and comments on a closely related work (see
Paraskevopoulos et al., 2013) which greatly improved many aspects of the
current paper as well. We are grateful to seminar participants at CREST,
Erasmus University, London School of Economics, Queen Mary University of
London, Imperial College, University of Essex, Birkbeck College University
of London, University of Nottingham, Cardiff University, University of
Manchester, Athens University of Economics and Business, and University of
Piraeus. We have also benefited from the comments given by participants (on
the closely related work) at the 3rd Humboldt-Copenhagen Conference on
Financial Econometrics (Humboldt University, Berlin, March 2013), the SNDE
21st Annual Symposium (University of Milan-Bicocca, March 2013), the 8th and
9th BMRC-QASS conferences on Macro and Financial Economics (Brunel
University, London, May 2013), the 7th CFE Conference (Senate House,
University of London, December 2013), and the 1st RASTANEWS Conference
(University of Milan-Bicocca, January 2014).}

$^{\dagger }${\footnotesize Address for correspondence: Menelaos Karanasos,
Economics and Finance, Brunel University, West London, UB3 3PH, UK; email:
menelaos.karanasos@brunel.ac.uk, tel: +44(0)1895265284, fax: +44
(0)1895269770.}
\end{abstract}

\newpage

\doublespacing

\section{Introduction}

The constancy of the parameters assumption made in the specification of time
series econometric models has been the subject of criticism for a long time.
\ It is argued that the assumption is inappropriate in the face of changing
institutions and a dynamically responding economic policy. These evolving
factors cause the parameter values characterizing economic relationships to
change over time. \ Partly to respond to the criticism and partly motivated
by the desire to construct dynamic models, econometricians have developed an
arsenal of powerful methods that attempt to capture the evolving nature of
our economy. \ Such frameworks include AR processes which contain multiple
abrupt breaks, and periodic and cyclical autoregressive models.

A methodology is presented in this paper for analyzing time varying systems
which is also applicable to the three aforementioned processes. A technique
is set forth for examining the periodic AR model, which overcomes the usual
requirement of expressing the periodic process in a vector AR (VAR) form.

The first attempts to develop theories for time varying models, made in the
1960's, were based on a recursive approach (Whittle, 1965) and on
evolutionary spectral representations (Abdrabbo and Priestley, 1967). \ Rao
(1970) used the method of weighted least squares to estimate an
autoregressive model with time dependent coefficients. Despite nearly half a
century of research work, the great advances, and the widely recognized
importance of time varying structures, the bulk of econometric models have
constant coefficients. There is a lack of a general theory that can be
employed to systematically explore their time series properties. Granger in
some of his last contributions highlighted the importance of the topic (see,
Granger 2007, and 2008).

There is a general agreement that the main obstacle to progress is the lack
of a universally applicable method yielding a closed form solution to
stochastic time varying difference equations. \ The present paper is part of
a research program aiming to produce and utilize closed form solutions to
AR\ processes with non stochastic time dependent coefficients. \ Our
methodology attempts to trace the path of these changing coefficients. To be
specific, in the time series literature, there is no method for finding the $%
p$ linearly independent solutions that we need to obtain the general
solution of the TV-AR model of order $p$. \ To keep the exposition tractable
and reveal its practical significance we work with low order specifications.

The main part of the paper begins with subsection 2.2, where we state the
second order difference equation with time variable coefficients, which is
our main object of inquiry. \ We start by writing this equation in a more
efficient way as an infinite linear system. \ The next step is to define the
matrix of coefficients, called the fundamental solution matrix, associated
with the system representation. \ This matrix is a workhorse of our research
and it is derived step by step from the time varying coefficients of the
difference equation.

The reader will have noticed that we have moved the goalposts, paradoxically
against us, from obtaining a solution for a time varying (low order)
difference equation, to solving an infinite linear system. \ The reason is
that the solution of such infinite systems has been made possible recently,
due to an extension of the standard Gaussian elimination, called the
infinite Gaussian elimination (see Paraskevopoulos, 2012). \ Applying this
infinite extension algorithm, we obtain the fundamental solutions, which
take explicit forms in terms of the determinants of the fundamental solution
matrix.

Subsection 2.3 contains the main theoretical result of the paper. \ Pursuing
the conventional route followed by the differential and difference equations
literature, we construct the general solution by finding its two parts, the
homogeneous one and a particular part. \ It is expressed as Theorem 1 and
its proof is in Appendix A. The coefficients in these solutions are
expressed as determinants of tridiagonal matrices. The second order
properties of the TV-AR process can easily be deduced from these solutions.
An additional benefit of these solutions is the facility with which linear
prediction can be produced. This allows us to provide a thorough description
of time varying models by deriving: first, multistep ahead forecasts, the
associated forecast error and the mean square error; second, the first two
unconditional moments of the process and its covariance structure. In
related works we provide results for the $p$ order and the more general
ascending order (see, for example, Paraskevopoulos et al., 2013). Our method
is a natural extension of the first order solution formula. It also includes
the linear difference equation with constant coefficients (see, for example,
Karanasos, 2001) as a special case.

The next two Sections of the paper, 3 and 4, apply our theoretical framework
to a few classic time series models, which are obvious candidates for a time
varying treatment. \ Linear systems with time dependent coefficients are not
only of interest in their own right, but, because of their connection with
periodic models and time series data which are subject to structural breaks.
They also provide insight into these processes as well. Viewing a periodic
AR (PAR) formulation as a TV model clearly obviates the need for VAR
analysis. For surveys and a review of some important aspects of PAR
processes see Franses (1996b), Franses and Paap (2004), Ghysels and Osborn
(2001), and Hurd and Miamee (2007). The authoritative studies by Osborn
(1988), Birchenhall et al. (1989), and Osborn and Smith (1989) applied these
models to consumption. del Barrio Castro and Osborn (2008) pointed out that
\textquotedblleft \textit{despite the attraction of PAR models from the
perspective of economic decision making in a seasonal context, the more
prominent approach of empirical workers is to assume that the AR
coefficients, except for the intercept, are constant over the seasons of the
year}\textquotedblright .\footnote{%
del Barrio Castro and Osborn (2008, 2012) (see the references therein for
this stream of important research; see also Taylor, 2002, 2003 and 2005)
test for seasonal unit roots in integrated PAR models.}

Despite the recognized importance of periodic processes for economics there
have been few attempts to investigate their time series properties (see,
among others, Franses, 1994, Franses, 1996a, Lund and Basawa, 2000, Franses
and Paap, 2005). Tiao and Grupe (1980) and Osborn (1991) analyzed these
models by converting them into a VAR process with constant coefficients. In
this paper we develop a general theory that can be employed to
systematically explore the fundamental properties of the periodic\
formulation. We remain within the univariate framework and we look upon the
PAR model as a stochastic difference equation with time varying (albeit
periodically varying) parameters. \ 

Although some theoretical analysis of periodic specifications was carried
out by the aforementioned studies the investigation of their fundamental
properties appears to have been limited to date. Cipra and Tlust\'{y}
(1987), Anderson and Vecchia (1993), Adams and Goodwin (1995), Shao (2008),
and Tesfaye et al. (2011) discuss parameter estimation and asymptotic
properties of periodic AR moving average (PARMA) specifications. Bentarzi
and Hallin (1994) and McLeod (1994) derive invertibility conditions and
diagnostic checks for such processes. Lund and Basawa (2000) develop a
recursive scheme for computing one-step ahead predictors for PARMA
specifications, and compute multi-step-ahead predictors recursively from the
one-step-ahead predictions. Anderson et al. (2013) develop a recursive
forecasting algorithm for periodic models. We derive explicit formulas that
allow the analytic calculation of the multi-step-ahead predictors.

We begin Section 3 with a PAR($2$) model. We limit our analysis to a low
order to save space and also since Franses (1996a) has documented that low
order PAR specifications often emerge in practice. \ First, we formulate it
as a TV model; then, we express its fundamental solution matrix as a block
Toeplitz matrix. \ This representation enables us to establish an explicit
formula for the general solution in terms of the determinant of such a block
matrix. The result is presented in Proposition \ref{PropGenSolPer}, which is
the equivalent to Theorem 1 with the incorporation of the seasonal effects.
\ That is, by taking account of seasons and periodicities, we obtain the
general solution, by constructing its homogeneous and particular parts and
then adding them up. \ In subsection 3.1, we turn our attention to a
different type of seasonality, namely the cyclical AR (CAR) model and we
provide its solution.

Section 4 is an application of the time varying framework to time series
subject to multiple structural breaks. \ We employ a technique analogous to
the one used in Section 3 on the PAR formulation. \ In particular, we
express the fundamental solutions of the AR($2$) model with $r$ abrupt
breaks, as determinants of block tridiagonal matrices. \ Again, we are able
to obtain the general solution by finding and adding the homogeneous and
particular solutions.

One of the advantages of our time varying framework is that we can trace the
entire path of the series under consideration. \ In Section 5, we employ
this information feature to derive the fundamental properties of the various
TV-AR processes. \ For example, simplified closed-form expressions of the
multi-step forecast error variances are derived for time series when low
order PAR models adequately describe the data. These formulae allow a fast
computation of the multi-step-ahead predictors. Finally, Section 6 concludes.

\section{Time Varying AR Models\label{SecTVAR}}

\subsection{Preliminaries and Purpose of Analysis}

\subsubsection{Notation}

Throughout the paper we adhere to the following conventions: ($\mathbb{Z}%
^{+} $) $\mathbb{Z},$ and ($\mathbb{R}^{+}$) $\mathbb{R}$ stand for the sets
of (positive) integers, and (positive) real numbers, respectively. Matrices
and vectors are denoted by upper and lower case boldface symbols,
respectively. For square matrices $\mathbf{X}=[x_{ij}]_{i,j=1,\ldots ,k}\in 
\mathbb{R}^{kxk}$ using standard notation, det$(\mathbf{X})$ or $\left \vert 
\mathbf{X}\right \vert $ denotes the determinant of matrix $\mathbf{X}$ and $%
adj(\mathbf{X})$ its adjoint matrix. To simplify our exposition we also
introduce the following notation: $t\in $ $\mathbb{Z}$, ($n,l$) $\in $ $%
\mathbb{Z}^{+}\times \mathbb{Z}^{+}$; $T=0,\ldots ,n$ denotes the `periods'
(i.e., years); $s=1,\ldots ,l$, denotes the `seasons' (i.e, quarters in a
year: $l=4$). The $t$ represents the present time and $k\in \mathbb{Z}^{+}$
the number of seasons such that at time $\tau _{k}=t-k$ information is given.

Let the triple $(\Omega ,\{ \tciFourier _{t},t\in 
\mathbb{Z}
\},P)$\ denote a complete probability space with a filtration, $\{
\tciFourier _{t}\}$,\ which is a non-decreasing sequence of $\sigma $-fields 
$\tciFourier _{t-1}\subseteq \tciFourier _{t}\subseteq \tciFourier $, $t\in 
\mathbb{Z}
$. The space of $P$-equivalence classes of finite complex random variables
with finite $p$-order is indicated by $L_{p}$. Finally, $H=L_{2}(\Omega
,\tciFourier _{t},P)$\ stands for a Hilbert space of random variables with
finite first and second moments. 

\subsubsection{The Problem}

The solution of the second order linear difference equation with non
constant coefficients is the building block for the extension of the well
known closed form solution of the first order to the $p$th order time
varying equation. As noted by Sydsaeter et al. (2008), in their classic text
(Further Mathematics for Economic Analysis, p. 403), in the case of second
order homogeneous linear difference equations with variable coefficients:

"\textit{There is no universally applicable method of discovering the two
linearly independent solutions that we need in order to find the general
solution of the equation.}"

We can identify two lines of inquiry that can be pursued to solve linear
difference equations with time varying coefficients. \ Searching for a
solution, one can follow either of the following two paths. The first is to
develop an analogous method to the standard one that exists for the linear $%
p $ order difference equation with constant coefficients: find the
eigenvalues, solve the characteristic equation, and obtain the closed form.
The second line of research searches for the generalization of the closed
form formula that exists for first order time varying difference equations.
\ Here, the way to proceed is to make up a conjecture and try to prove it by
induction. The two strands of the literature have taken important steps, but
have not provided us with a general solution method that we can apply; the
existing results lack generality and applicability. \ To be more specific,
the research problem we face is that there is a lack of a universally
applicable method yielding a closed form solution to stochastic higher order
difference equations with time dependent coefficients.

A general method for solving infinite linear systems with row-finite
coefficient matrices has recently been established by Paraskevopoulos
(2012). It is a modified version of the standard Gauss-Jordan elimination
method implemented under a right pivot strategy, called infinite
Gauss-Jordan elimination. Expressing the linear difference equation of
second order with time dependent coefficients as an infinite linear system,
the Gaussian elimination part of the method is directly applicable. It
generates two linearly independent homogeneous solution sequences. The
general term of each solution sequence turns out to be a continuant
determinant. The general solutions of the homogeneous and nonhomogeneous
difference equation are expressible as a single Hessenbergian, that is, a
determinant of a lower Hessenberg matrix (see Karanasos, Paraskevopoulos and
Dafnos 2013). Theorem $3$ in Paraskevopoulos et al. (2013) affords an easy
means of finding, for a given lower Hessenberg matrix, its ordinary
expansion in non-determinant form (see also Paraskevopoulos and Karanasos,
2013). These results are extendible to the solution of the $p$th and
ascending order time varying linear difference equations in terms of a
single Hessenbergian (see Paraskevopoulos et al., 2013). This makes it
possible to introduce, in the above cited reference, a unified theory for
time varying models.

\subsection{Fundamental Solution Matrices\label{SubSecFundSol}}

The main theoretical contribution of this Section is the development of a
method that provides the closed form of the general solution to a TV-AR($2$)
model.

Next we give the main definition that we will use in the rest of the paper.
Consider a second order stochastic difference equation with time dependent
coefficients, which is equivalent to the time varying AR($2$) process, given
by 
\begin{equation}
y_{t}=\phi _{0}(t)+\phi _{1}(t)y_{t-1}+\phi _{2}(t)y_{t-2}+\varepsilon _{t}%
\text{, }  \label{TVAR(P)}
\end{equation}%
where $\{ \varepsilon _{t},t\in 
\mathbb{Z}
\}$ is a sequence of zero mean serially uncorrelated random variables
defined on $L_{2}(\Omega ,\tciFourier _{t},P)$ with $\mathbb{E}[\varepsilon
_{t}\left \vert \tciFourier _{t-1}\right. ]=0$ a.s., and finite variance: $%
0<M_{l}<\sigma _{t}^{2}<M<\infty $, $\forall $ $t$, ($M_{l},M$) $\in \mathbb{%
R}^{+}\times \mathbb{R}^{+}$.

\begin{remark}
\label{RemHeter}We have relaxed the assumption of homoscedasticity (see
also, among others, Paraskevopoulos et al., 2013 and Karanasos et al.,
2013), which is likely to be violated in practice and allow $\varepsilon
_{t} $ to follow, for example, a periodical GARCH type of process (see,
Bollerslev and Ghysels, 1996).
\end{remark}

The fundamental solution sequence, and in general all the solution
sequences, must necessarily be functions of the independent variable $t$, so
as to satisfy eq. (\ref{TVAR(P)}). Our intermediate objective is to obtain
the fundamental solution matrix, denoted below by $\mathbf{\Phi }_{t,k}$,
which is associated with our stochastic difference equation (\ref{TVAR(P)});
the $\mathbf{\Phi }_{t,k}$ matrix will be derived from the time varying
coefficients of eq. (\ref{TVAR(P)}). \ The best way to appreciate the
representation of the fundamental solution matrix is to view the stochastic
difference equation as a linear system. \ We carry out this construction
below. \ Once we have this stepping stone in place, then we can pursue our
ultimate objective, by computing the determinants of the $\mathbf{\Phi }%
_{t,k}$, which will give us the linearly independent solutions sequences to
the difference equation.

Equation (\ref{TVAR(P)}) written as 
\begin{equation}
\phi _{2}(t)y_{t-2}+\phi _{1}(t)y_{t-1}-y_{t}=-[\phi _{0}(t)+\varepsilon
_{t}],  \label{Difference(P)}
\end{equation}%
takes the infinite row (and column)-finite system form 
\begin{equation}
\mathbf{\Phi \cdot y}=-\mathbf{\phi }-\mathbf{\varepsilon ,}
\label{PHIMatrix}
\end{equation}%
where

\begin{equation*}
\mathbf{\Phi =}\left( 
\begin{array}{ccccccc}
\phi _{2}(\tau _{k}+1) & \phi _{1}(\tau _{k}+1) & -1 & 0 & 0 & 0 & ... \\ 
0 & \phi _{2}(\tau _{k}+2) & \phi _{1}(\tau _{k}+2) & -1 & 0 & 0 & ... \\ 
0 & 0 & \phi _{2}(\tau _{k}+3) & \phi _{1}(\tau _{k}+3) & -1 & 0 & ... \\ 
\vdots & \vdots & \vdots & \vdots & \vdots & \vdots & \vdots \vdots \vdots%
\end{array}%
\right) ,
\end{equation*}%
(row-finite is an infinite matrix whose rows have finite non zero elements)
and%
\begin{equation*}
\mathbf{y=}\left( 
\begin{array}{l}
y_{\tau _{k}-1} \\ 
y_{\tau _{k}} \\ 
y_{\tau _{k}+1} \\ 
y_{\tau _{n}+2} \\ 
y_{\tau _{k}+3} \\ 
y_{\tau _{k}+4} \\ 
\multicolumn{1}{c}{\vdots}%
\end{array}%
\right) \text{, }\mathbf{\phi }\mathbb{=}\left( 
\begin{array}{l}
\phi _{0}(\tau _{k}+1) \\ 
\phi _{0}(\tau _{k}+2) \\ 
\phi _{0}(\tau _{k}+3) \\ 
\multicolumn{1}{c}{\vdots}%
\end{array}%
\right) \text{, }\mathbf{\varepsilon =}\left( 
\begin{array}{l}
\varepsilon _{_{\tau _{k}+1}} \\ 
\varepsilon _{_{\tau _{k}+2}} \\ 
\varepsilon _{_{\tau _{k}+3}} \\ 
\multicolumn{1}{c}{\vdots}%
\end{array}%
\right)
\end{equation*}%
(recall that $\tau _{k}=t-k$). The system representation results from the
values that the coefficients take in successive time periods. The
equivalence of (\ref{Difference(P)}) and (\ref{PHIMatrix}) follows from the
fact that the $i$th equation in (\ref{PHIMatrix}), as a result of the
multiplication of the $i$th row of $\mathbf{\Phi }$ by the column of $y$%
{\footnotesize s} equated to $-[\phi _{0}(\tau _{k}+i)+\varepsilon _{_{\tau
_{k}+i}}]$, is equivalent to eq. (\ref{Difference(P)}), as of time $\tau
_{k}+i$. The $\mathbf{\Phi }$ matrix in eq. (\ref{PHIMatrix}) can be
partitioned as 
\begin{equation*}
\mathbf{\Phi =}\left( 
\begin{tabular}{l|l}
$\mathbf{P}$ & $\mathbf{C}$%
\end{tabular}%
\right) ,
\end{equation*}%
where%
\begin{equation*}
\mathbf{P=}\left( 
\begin{array}{cc}
\phi _{2}(\tau _{k}+1) & \phi _{1}(\tau _{k}+1) \\ 
0 & \phi _{2}(\tau _{k}+2) \\ 
0 & 0 \\ 
\vdots & \vdots%
\end{array}%
\right) ,\text{ }\mathbf{C=}\left( 
\begin{array}{ccccc}
-1 & 0 & 0 & 0 & ... \\ 
\phi _{1}(\tau _{k}+2) & -1 & 0 & 0 & ... \\ 
\phi _{2}(\tau _{k}+3) & \phi _{1}(\tau _{k}+3) & -1 & 0 & ... \\ 
\vdots & \vdots & \vdots & \vdots & \vdots \vdots \vdots%
\end{array}%
\right) .
\end{equation*}%
That is, $\mathbf{P}$ consists of the first $2$ columns of $\mathbf{\Phi }$
and the $j$th column of $\mathbf{C}$, $j=1,2,\ldots ,$ is the ($2+j$)th
column of $\mathbf{\Phi }$. We will denote the $2$nd column of the $k\times
2 $ top submatrix of the matrix $\mathbf{P}$ by $\mathbf{\phi }_{t,k}$:%
\begin{equation*}
(\mathbf{\phi }_{t,k})^{\prime }=\left( 
\begin{array}{lllll}
{\small \phi }_{1}{\small (\tau }_{k}{\small +1),} & {\small \phi }_{2}%
{\small (\tau }_{k}{\small +2),} & 0, & \ldots & ,0%
\end{array}%
\right) .
\end{equation*}

The $k\times (k-1)$ top submatrix of matrix $\mathbf{C}$ is called the core
solution matrix and is denoted as 
\begin{equation}
\mathbf{C}_{t,k}=\left( 
\begin{array}{ccccc}
-1 &  &  &  &  \\ 
\phi _{1}(\tau _{k}+2) & -1 &  &  &  \\ 
\phi _{2}(\tau _{k}+3) & \phi _{1}(\tau _{k}+3) & -1 &  &  \\ 
& \ddots & \ddots & \ddots &  \\ 
&  & \phi _{2}(t-1) & \phi _{1}(t-1) & -1 \\ 
&  &  & \phi _{2}(t) & \phi _{1}(t)%
\end{array}%
\right)  \label{CoreMatrix}
\end{equation}%
(here and in what follows empty spaces in a matrix have to be replaced by
zeros). The fundamental solution matrix is obtained from the core solution
matrix $\mathbf{C}_{t,k}$, augmented on the left by the $\mathbf{\phi }%
_{t,k} $ column. That is,

\begin{eqnarray}
\mathbf{\Phi }_{t,k} &=&\left( 
\begin{tabular}{l|l}
$\mathbf{\phi }_{t,k}$ & $\mathbf{C}_{t,k}$%
\end{tabular}%
\right) =  \notag \\
&&\left( 
\begin{array}{cccccc}
\phi _{1}(\tau _{k}+1) & -1 &  &  &  &  \\ 
\phi _{2}(\tau _{k}+2) & \phi _{1}(\tau _{k}+2) & -1 &  &  &  \\ 
& \phi _{2}(\tau _{k}+3) & \phi _{1}(\tau _{k}+3) & -1 &  &  \\ 
&  & \ddots & \ddots & \ddots &  \\ 
&  &  & \phi _{2}(t-1) & \phi _{1}(t-1) & -1 \\ 
&  &  &  & \phi _{2}(t) & \phi _{1}(t)%
\end{array}%
\right) ,  \label{FAIMAR(p)X}
\end{eqnarray}%
(recall that $\tau _{k}=t-k$). Formally $\mathbf{\Phi }_{t,k}$ is a square $%
k\times k$ matrix whose ($i,j$) entry $1\leq i,j\leq k$ is given by

\begin{equation*}
\left \{ 
\begin{array}{cccccc}
-1 & \text{if} &  & i=j-1, & \text{and} & 2\leq j\leq k, \\ 
\phi _{1+m}(t-k+i) & \text{if} & m=0,1, & i=j+m, & \text{and} & 1\leq j\leq
k-m, \\ 
0 &  & \text{otherwise.} &  &  & 
\end{array}%
\right. 
\end{equation*}

It is a continuant or tridiagonal matrix, that is a matrix that is both an
upper and lower Hessenberg matrix. We may also characterize it as a `time
varying' Toeplitz matrix, because its time invariant version is a Toeplitz
matrix of bandwidth $3$. We next define the bivariate function $\xi :\mathbb{%
Z}\times \mathbb{Z}^{+}\longmapsto \mathbb{R}$ by 
\begin{equation}
\xi _{t,k}=\text{det}(\mathbf{\Phi }_{t,k})  \label{KSIAR(2)}
\end{equation}%
coupled with the initial values $\xi _{t,0}=1$, and $\xi _{t,-1}=0$. That
is, $\xi _{t,k}$ for $k\geq 2$, is a determinant of a $k\times k$ matrix;
each of the two nonzero diagonals (below the superdiagonal) of this matrix
consists of the time varying coefficients $\phi _{m}(\cdot )$, $m=1,2$, from 
$t-k+m$ to $t$. In other words, $\xi _{t,k}$ is a\ $k$th-order tridiagonal
determinant. Paraskevopoulos and Karanasos (2013) give its ordinary
expansion in non-determinant form (a closed form solution).

\subsection{Main Theorem\label{SubSecMainTheorem}}

This short section contains the statement of our main theorem.

\begin{theorem}
\label{TheoGenSol}The general solution of eq. (\ref{TVAR(P)}) with free
constants (initial condition values) $y_{t-k}$, $y_{t-k-1}$ is given by 
\begin{equation}
y_{t,k}^{gen}=y_{t,k}^{hom}+y_{t,k}^{par},  \label{TVAR(p)SOL}
\end{equation}%
where 
\begin{eqnarray*}
y_{t,k}^{hom} &=&\xi _{t,k}y_{t-k}+\phi _{2}(t-k+1)\xi _{t,k-1}y_{t-k-1}, \\
y_{t,k}^{par} &=&\sum_{i=0}^{k-1}\xi _{t,i}\phi
_{0}(t-i)+\sum_{i=0}^{k-1}\xi _{t,i}\varepsilon _{t-i}.
\end{eqnarray*}
\end{theorem}

In the above Theorem $y_{t,k}^{gen}$\ is decomposed into two parts: first,
the $y_{t,k}^{hom}$\ part, which is written in terms of the two\ free
constants ($y_{t-k-m}$, $m=0,1$), and, second, the $y_{t,k}^{par}$\ part,
which contains the\ time varying drift terms and the error terms from time $%
t-k+1$\ to time $t$.

Notice that the `coefficients' of eq. (\ref{TVAR(p)SOL}), that is, the $\xi $%
's are expressed as continuant determinants. Moreover, for `$k=0$' (for $i>j$
we use the convention $\sum_{q=i}^{j}(\cdot )=0$), since $\xi _{t,0}=1$ and $%
\xi _{t,-1}=0$ (see eq. (\ref{KSIAR(2)})), eq. (\ref{TVAR(p)SOL}) becomes an
`identity': $y_{t,0}^{gen}=y_{t}$. Similarly, when `$k=1$' eq. (\ref%
{TVAR(p)SOL}), since $\xi _{t,1}=\phi _{1}(t)$ and $\xi _{t,0}=1$, reduces
to $y_{t,1}^{gen}=\phi _{1}(t)y_{t-1}+\phi _{2}(t)y_{t-2}+\phi
_{0}(t)+\epsilon _{t}$.

In the next Section, we illustrate the above claims in the context of a
simple seasonal process with fixed periodicity, and a cyclical model as well.

\section{Periodic AR(2) Model\label{SecPAR(2)}}

Periodic regularities are phenomena occurring at the same season every year,
so analogous to each other that we can view them as recurrences of the same
event. \ Many economic time series are periodic in this sense. \ In the
present Section we express them in a mathematical model, so that we can then
employ it for forecasting and control. Gladyshev (1961) introduced a
technique which still dominates the literature. \ He begins by decomposing
the series into subperiods; then he treats each point within a subperiod as
one part of a multivariate process. \ In this way he transforms a univariate
non-stationary formulation into a multivariate stationary one. Following
Gladyshev, Tiao and Grupe (1980) and Osborn (1991) treated periodic
autoregressions as conventional nonperiodic VAR processes. But, as pointed
out by Lund et al. (2006), even low order specifications can have an
inordinately large numbers of parameters. A PAR($1$) model for daily data,
for example, has $365$ autoregressive parameters. Its time invariant VAR
form will contain $365$ variables, and this is a handicap, especially for
forecasting.

The most common case is the modeling in one dimensional time repetition at
equal intervals. \ In this Section we present a re-examination of the
periodic modeling problem. \ Our approach differs from most of the existing
literature in that we stay within the univariate framework.

A periodic AR model of order $2$ with $l$ seasons, PAR($2;l$), is defined as%
\begin{equation}
y_{t_{s}}=\phi _{0,s}+\phi _{1,s}y_{t_{s}-1}+\phi
_{2,s}y_{t_{s}-2}+\varepsilon _{t_{s}}\text{ }  \label{PAR(2)}
\end{equation}%
where $t_{s}=Tl+s$, $s=1,\ldots ,l$, that is time $t_{s}$ is at the $s$th
season and $\phi _{m,s}$, $m=1,2$, are the periodically (or seasonally)
varying autoregressive coefficients. For example, if $s=l$ (that is, we are
at the $l$th season) then the periodically varying coefficients are $\phi
_{m,l}$ whereas if $s=1$ (that is, we are at the $1$st season) then the
periodically varying coefficients are $\phi _{m,1}$; $\phi _{0,s}$ is a
periodically varying drift%
. The above process nests the AR($2$) model as a special case
if we assume that the drift and all the AR parameters are constant, that is: 
$\phi _{m,s}=\varphi _{m}$, $m=0,1,2$, for all $t$.

The PAR($2;l$) model can be expressed as the time varying AR($2$) model in
eq.(\ref{TVAR(P)}): 
\begin{equation*}
y_{t}=\phi _{0}(t)+\phi _{1}(t)y_{t-1}+\phi _{2}(t)y_{t-2}+\varepsilon _{t}%
\text{,}
\end{equation*}%
where $\phi _{m}(t)=\phi _{m}(\tau _{Tl})$, $m=0,1,2$, $\tau _{Tl}=t-Tl$,
are the periodically (or seasonally) varying autoregressive coefficients: $%
\phi _{m,s}\triangleq \phi _{m}(Tl+s)$, $s=1,\ldots ,l$.

For the PAR($2;l$) model the continuant matrix $\mathbf{\Phi }_{t,nl}$ in
eq. (\ref{FAIMAR(p)X}) (we assume that information is given at time $\tau
_{nl}=t-nl$ for ease of exposition; it can of course be given at any time $%
\tau _{nl+s}=t-nl-s$) can be expressed as a block Toeplitz matrix. Thus, we
have

\begin{equation}
\xi _{t,nl}=\left \vert \mathbf{\Phi }_{t,nl}\right \vert ,
\label{BlockToeplitz}
\end{equation}%
with%
\begin{equation*}
\mathbf{\Phi }_{t,nl}=\left( 
\begin{array}{lllll}
\mathbf{\Phi }_{\tau _{(n-1)l},l} & \overline{\mathbf{0}} &  &  &  \\ 
\widetilde{\mathbf{0}}_{\tau _{n-2}} & \mathbf{\Phi }_{\tau _{(n-2)l},l} & 
\overline{\mathbf{0}} &  &  \\ 
& \ddots & \ddots & \ddots &  \\ 
&  & \widetilde{\mathbf{0}}_{\tau _{1}} & \mathbf{\Phi }_{\tau _{l},l} & 
\overline{\mathbf{0}} \\ 
&  &  & \widetilde{\mathbf{0}}_{t} & \mathbf{\Phi }_{t,l}%
\end{array}%
\right) ,
\end{equation*}%
where $\overline{\mathbf{0}}$ is an $l\times l$ matrix of zeros except for $%
-1$ in its ($l,1$) entry; $\widetilde{\mathbf{0}}_{t}$ is an $l\times l$
matrix of zeros except $\phi _{2}(t-l+1)$, in its ($1,l$) entry. \ Since $%
\phi _{m}(\tau _{Tl})=\phi _{m}(t)$: $\widetilde{\mathbf{0}}_{\tau _{Tl}}=%
\widetilde{\mathbf{0}}_{t}$ and $\mathbf{\Phi }_{\tau _{Tl},l}=\mathbf{\Phi }%
_{t,l}$. Thus the block diagonal matrix $\mathbf{\Phi }_{t,nl}$ can be
written as%
\begin{equation}
\mathbf{\Phi }_{t,nl}=\left( 
\begin{array}{lllll}
\mathbf{\Phi }_{t,l} & \overline{\mathbf{0}} &  &  &  \\ 
\widetilde{\mathbf{0}}_{t} & \mathbf{\Phi }_{t,l} & \overline{\mathbf{0}} & 
&  \\ 
& \ddots & \ddots & \ddots &  \\ 
&  & \widetilde{\mathbf{0}}_{t} & \mathbf{\Phi }_{t,l} & \overline{\mathbf{0}%
} \\ 
&  &  & \widetilde{\mathbf{0}}_{t} & \mathbf{\Phi }_{t,l}%
\end{array}%
\right) ,  \label{Block Toeplitz1}
\end{equation}%
where $\mathbf{\Phi }_{t,l}$ is the continuant or tridiagonal matrix $%
\mathbf{\Phi }_{t,k}$ matrix defined in eq. (\ref{FAIMAR(p)X}) when $k=l$.
That is 
\begin{equation}
\mathbf{\Phi }_{t,l}=\left( 
\begin{array}{cccccc}
\phi _{1}(\tau _{l}+1) & -1 &  &  &  &  \\ 
\phi _{2}(\tau _{l}+2) & \phi _{1}(\tau _{l}+2) & -1 &  &  &  \\ 
& \phi _{2}(\tau _{l}+3) & \phi _{1}(\tau _{l}+3) & -1 &  &  \\ 
&  & \ddots & \ddots & \ddots &  \\ 
&  &  & \phi _{2}(t-1) & \phi _{1}(t-1) & -1 \\ 
&  &  &  & \phi _{2}(t) & \phi _{1}(t)%
\end{array}%
\right) .  \label{FAIM1}
\end{equation}

\begin{proposition}
\label{PropGenSolPer}The general solution of eq. (\ref{PAR(2)}) with free
constants (initial condition values) $y_{t-nl}$, $y_{t-nl-1}$ is given by 
\begin{equation}
y_{t,nl}^{gen}=y_{t,nl}^{hom}+y_{t,nl}^{par},
\end{equation}%
where 
\begin{eqnarray*}
y_{t,nl}^{hom} &=&\xi _{t,nl}y_{t-nl}+\phi _{2}(t-nl+1)\xi
_{t,nl-1}y_{t-nl-1}, \\
y_{t,nl}^{par} &=&\sum_{s=0}^{l-1}\sum_{T=0}^{n-1}\xi _{t,Tl+s}\phi
_{0}(t-s)+\sum_{i=0}^{nl-1}\xi _{t,i}\varepsilon _{t-i},
\end{eqnarray*}%
and $\xi _{t,nl}$ is given either in eq. (\ref{BlockToeplitz}) or in
Proposition (\ref{ProPAR(2;l)KSI}) in Appendix \ref{SecAppendABAR}.
\end{proposition}

The proof of the above Proposition follows immediately from Theorem 1 and
the definition of the periodic model (\ref{PAR(2)}).

\subsection{Cyclical AR($2$) process\label{SubSecCAR(2)}}

Some economic series exhibit oscillations which are not associated with the
same fixed period every year. Despite their lack of fixed periodicity, such
time series are predictable to a certain degree.

Rather than setting up a general model from first principles, we
re-interpret the periodic model with some modifications. \ In particular, we
now assume that we have $d+1$ cycles, with $0\leq d\leq l-1$. \ Then, $%
s_{j}=l_{j-1}+1,\ldots ,l_{j}$, $j=1,\ldots ,d+1$ (with $0=l_{0}<l_{1}<%
\ldots <l_{d}<l_{d+1}=l$) are the seasons in cycle $j$. Thus we can write $%
\phi _{m,s_{j}}\triangleq \phi _{m}(t_{s_{j}})$, $m=0,1,2$, $%
t_{s_{j}}=Tl+s_{j}$. A CAR($2$) model with $l$ seasons and $d+1$ cycles (CAR(%
$2;l;d$)) is defined as%
\begin{equation}
y_{t_{s_{j}}}=\phi _{0,s_{j}}+\phi _{1,s_{j}}y_{t_{s_{j}}-1}+\phi
_{2,s_{j}}y_{t_{s_{j}}-2}+\varepsilon _{t_{s_{j}}}\text{.}
\label{CAR(2;l;d)}
\end{equation}

For the above process, $\mathbf{\Phi }_{t,l}$ in eq. (\ref{FAIM1}) can be
written as

\begin{equation}
\mathbf{\Phi }_{t,l}=\left[ 
\begin{array}{lllll}
\mathbf{\Phi }_{t-l_{d},l_{d+1}-l_{d}} & \overline{\mathbf{0}}_{d} &  &  & 
\\ 
\widetilde{\mathbf{0}}_{d} & \mathbf{\Phi }_{t-l_{d-1},l_{d}-l_{d-1}} & 
\overline{\mathbf{0}}_{d-1} &  &  \\ 
& \ddots  & \ddots  & \ddots  &  \\ 
&  & \widetilde{\mathbf{0}}_{2} & \mathbf{\Phi }_{t-l_{1},l_{2}-l_{1}} & 
\overline{\mathbf{0}}_{1} \\ 
&  &  & \widetilde{\mathbf{0}}_{1} & \mathbf{\Phi }_{t,l_{1}}%
\end{array}%
\right] ,  \label{PHI(l)CAR2}
\end{equation}%
where first, the $j$\ ($j=1,\ldots ,d+1$) block of the main diagonal is $%
\mathbf{\Phi }_{t-l_{j-1},l_{j}-l_{j-1}}$, that is a $(l_{j}-l_{j-1})\times
(l_{j}-l_{j-1})$\ banded `time varying' Toeplitz matrix of bandwidth $3$:%
\begin{equation*}
\mathbf{\Phi }_{t-l_{j-1},l_{j}-l_{j-1}}=\left( 
\begin{array}{cccccc}
\phi _{1}(\tau _{l_{j}}+1) & -1 &  &  &  &  \\ 
\phi _{2}(\tau _{l_{j}}+2) & \phi _{1}(\tau _{l_{j}}+2) & -1 &  &  &  \\ 
& \phi _{2}(\tau _{l_{j}}+3) & \phi _{1}(\tau _{l_{j}}+3) & -1 &  &  \\ 
&  & \ddots  & \ddots  & \ddots  &  \\ 
&  &  & \phi _{2}(\tau _{l_{j-1}}-1) & \phi _{1}(\tau _{l_{j-1}}-1) & -1 \\ 
&  &  &  & \phi _{2}(\tau _{l_{j-1}}) & \phi _{1}(\tau _{l_{j-1}})%
\end{array}%
\right) ,
\end{equation*}%
second, the $j$\ ($j=1,\ldots ,d$) block of the subdiagonal, $\widetilde{%
\mathbf{0}}_{j}$, is a $(l_{j}-l_{j-1})\times (l_{j+1}-l_{j})$\ matrix of
zeros except for $\phi _{2}(\tau _{l_{j}}+1)$\ in its $1\times
(l_{j+1}-l_{j})$\ entry, and third, the $j$\ block of the superdiagonal $%
\overline{\mathbf{0}}_{j}$, is a $(l_{j+1}-l_{j})\times (l_{j}-l_{j-1})$
matrix of zeros except for $-1$ in its $(l_{j+1}-l_{j})\times 1$ entry, and
iv) there are zeros elsewhere.

\section{Abrupt Breaks}

Our general result has been presented in Section \ref{SubSecMainTheorem}. In
the current Section, we discuss still another example in order\ to both make
our analysis clearer and to demonstrate its applicability. One important
case is that of $r$, $0\leq r\leq k-1$, abrupt breaks at times $t-k_{1}$, $%
t-k_{2}$, $\ldots $ , $t-k_{r}$, where $0=k_{0}<k_{1}<k_{2}<\cdots
<k_{r}<k_{r+1}=k$, $k_{r}\in \mathbb{Z}^{+}$, and $k_{r}$ is finite. That
is, between $t-k=t-k_{r+1}$\ and the present time $t=t-k_{0}$ the AR($2$)
process contains $r$ structural breaks and the switch from one set of
parameters to another is abrupt. In particular%
\begin{equation}
y_{\tau }=\phi _{0,j}+\phi _{1,j}y_{\tau -1}+\phi _{2,j}y_{\tau -2}+\sigma
_{j}^{2}e_{\tau ,j},  \label{ABAR(2)}
\end{equation}%
for $\tau =t-k_{j-1},\ldots ,t-k_{j}+1,j=1,\ldots ,r+1$ and $e_{t,j}\sim $
i.i.d $(0,1)$ $\forall $ $t,j$. Within the class of AR($2$) processes, this
specification is quite general and allows for intercept and slope shifts as
well as changes in the error variances (see also Pesaran et al., 2006). Each
regime $j$\ is characterized by a vector of autoregressive coefficients: $%
\phi _{0,j}$, $\mathbf{\phi }_{j}^{\prime }=(\phi _{1,j},\phi _{2,j})$, and
an error term variance, $0<\sigma _{j}^{2}<M_{j}<\infty $ $\forall $ $j$, $%
M_{j}\in \mathbb{R}^{+}$. We term this model abrupt breaks AR process of
order ($2;r$) (ABAR($2;r$)).

For the AR($2$) model with $r$ abrupt breaks, $\xi _{t,k}$\ in eq. (\ref%
{KSIAR(2)}) can be written as the determinant of a partitioned (or a block)
tridiagonal matrix 
\begin{equation}
\xi _{t,k}=\det (\mathbf{\Phi }_{t,k})=\left \vert 
\begin{array}{lllll}
\mathbf{\Phi }_{t-k_{r},k_{r+1}-k_{r}} & \overline{\mathbf{0}}_{r} &  &  & 
\\ 
\widetilde{\mathbf{0}}_{r} & \mathbf{\Phi }_{t-k_{r-1},k_{r}-k_{r-1}} & 
\overline{\mathbf{0}}_{r-1} &  &  \\ 
& \ddots  & \ddots  & \ddots  &  \\ 
&  & \widetilde{\mathbf{0}}_{2} & \mathbf{\Phi }_{t-k_{1},k_{2}-k_{1}} & 
\overline{\mathbf{0}}_{1} \\ 
&  &  & \widetilde{\mathbf{0}}_{1} & \mathbf{\Phi }_{t,k_{1}}%
\end{array}%
\right \vert ,  \label{KSIABAR2}
\end{equation}%
where first, the $j$\ ($j=1,\ldots ,r+1$) block of the main diagonal is $%
\mathbf{\Phi }_{t-k_{j-1},k_{j}-k_{j-1}}$, \newline
that is a $(k_{j}-k_{j-1})\times (k_{j}-k_{j-1})$\ banded Toeplitz matrix of
bandwidth $3$:%
\begin{equation*}
\mathbf{\Phi }_{t-k_{j-1},k_{j}-k_{j-1}}=\left( 
\begin{array}{ccccc}
\phi _{1,j} & -1 &  &  &  \\ 
\phi _{2,j} & \phi _{1,j} & -1 &  &  \\ 
& \ddots  & \ddots  & \ddots  &  \\ 
&  & \phi _{2,j} & \phi _{1,j} & -1 \\ 
&  &  & \phi _{2,j} & \phi _{1,j}%
\end{array}%
\right) ,
\end{equation*}%
with $\xi _{t-k_{j-1},k_{j}-k_{j-1}}=\left \vert \mathbf{\Phi }%
_{t-k_{j-1},k_{j}-k_{j-1}}\right \vert =\frac{1}{\lambda _{1,j}-\lambda _{2,j}%
}(\lambda _{1,j}^{k_{j}-k_{j-1}+1}-\lambda _{2,j}^{k_{j}-k_{j-1}+1})$, and
the second equality holds if and only if $\lambda _{1,j}\neq \lambda _{2,j}$
(where $1-\phi _{1,j}B-\phi _{2,j}B^{2}=(1-\lambda _{1,j}B)(1-\lambda
_{2,j}B)$), second, the $j$\ ($j=1,\ldots ,r$) block of the subdiagonal, $%
\widetilde{\mathbf{0}}_{j}$, is a $(k_{j}-k_{j-1})\times (k_{j+1}-k_{j})$\
matrix of zeros except for $\phi _{2,j}$\ in its $1\times (k_{j+1}-k_{j})$\
entry, and third, the $j$\ block of the superdiagonal $\overline{\mathbf{0}}%
_{j}$, is a $(k_{j+1}-k_{j})\times (k_{j}-k_{j-1})$ matrix of zeros except
for $-1$ in its $(k_{j+1}-k_{j})\times 1$ entry, and iv) there are zeros
elsewhere.

\begin{corollary}
\label{CorGenSolStrBreak}The general solution of the ABAR($2;r$) model in
eq. (\ref{ABAR(2)}) with free constants (initial condition values) $y_{t-k}$%
, $y_{t-k-1}$, is given by 
\begin{equation*}
y_{t,k}^{gen}=y_{t,k}^{hom}+y_{t,k}^{par},
\end{equation*}%
where%
\begin{eqnarray*}
y_{t,k}^{hom} &=&\xi _{t,k}y_{t-k}+\phi _{2}(t-k+1)\xi _{t,k-1}y_{t-k-1}, \\
y_{t,k}^{par} &=&\sum_{j=1}^{r+1}\phi _{0,j}\sum_{i=k_{j-1}}^{k_{j}-1}\xi
_{t,i}+\sum_{j=1}^{r+1}\sigma _{j}^{2}\sum_{i=k_{j-1}}^{k_{j}-1}\xi
_{t,i}e_{t-i,j},
\end{eqnarray*}%
and $\xi _{t,k}$ is given either in eq. (\ref{KSIABAR2}) or in Proposition %
\ref{ProABAR(2,n)KSI} (see Appendix B).%
\end{corollary}

The proof of the above Corollary follows immediately from Theorem 1 and the
definition of the ABAR($2;r$) model in eq. (\ref{ABAR(2)}).

\section{Prediction and Moment Structure\label{SecSecMom}}

We turn our attention to the fundamental properties of the various TV-AR($2$%
) processes. Armed with a powerful technique for manipulating time varying
models we may now provide a thorough description of the processes (\ref%
{TVAR(P)}) by deriving, first, its multistep ahead predictor, the associated
forecast error and the mean square error; second, the first two
unconditional moments of this process, and third, its covariance structure.

\subsection{Multi Step Forecasts}

Taking the conditional expectation of eq. (\ref{TVAR(p)SOL}) with respect to
the $\sigma $ field $\tciFourier _{\tau _{k}}$ ($\tau _{k}=t-k$) yields the
following Proposition.

\begin{proposition}
\label{ProOptPred}For the TV-AR($2$) model the $k$-step-ahead optimal (in $%
L_{2}$-sense) linear predictor of $y_{t}$, $\mathbb{E}(y_{t}\left \vert
\tciFourier _{\tau _{k}}\right. )$,\ is readily seen to be%
\begin{equation}
\mathbb{E}(y_{t}\left \vert \tciFourier _{\tau _{k}}\right.
)=\sum_{i=0}^{k-1}\xi _{t,i}\phi _{0}(t-i)+\xi _{t,k}y_{t-k}+\phi
_{2}(t-k+1)\xi _{t,k-1}y_{t-k-1}.  \label{CE-TVAR(P)}
\end{equation}%
In addition, the forecast error for the above $k$-step-ahead predictor, $%
\mathbb{FE}(y_{t}\left \vert \tciFourier _{\tau _{k}}\right. )=y_{t}-\mathbb{E%
}[y_{t}\left \vert \tciFourier _{\tau _{k}}\right. ]$, is given by%
\begin{equation}
\mathbb{FE}(y_{t}\left \vert \tciFourier _{\tau _{k}}\right. )=\Xi
_{t,k}(B)\varepsilon _{t}=\sum_{i=0}^{k-1}\xi _{t,i}B^{i}\varepsilon _{t},
\label{FE-TVAR(P)}
\end{equation}%
and it is expressed in terms of $k$ error terms from time $t-k+1$\ to time $t
$; the coefficient of the error term at time $t-i$, $\xi _{t,i}$, is the
determinant of an $i\times i$ matrix ($\Phi _{t,i}$), each nonzero variable
diagonal of which consists of the AR time varying coefficients $\phi _{m}(%
\mathfrak{\cdot })$, $m=1,2$ from time $t-i+m$\ to $t$. \newline
The mean square error is%
\begin{equation}
\mathbb{V}ar[\mathbb{FE(}y_{t}\left \vert \tciFourier _{\tau _{k}}\right.
)]=\Xi _{t,k}^{(2)}(B)\sigma _{t}^{2}=\sum_{i=0}^{k-1}\xi _{t,i}^{2}B^{i}%
\mathbb{\sigma }_{t}^{2},  \label{VFE-TVAR(P)}
\end{equation}%
which is expressed in terms of $k$\ variances from time $t-k+1$\ to time $t$%
, with time varying coefficients (the squared $\xi ${\footnotesize s})%
.
\end{proposition}

The following Corollary presents results for the forecasts from PAR and CAR
processes.

\begin{corollary}
For the PAR($2;l$) and the CAR($2;l;d$) models (see eqs. (\ref{PAR(2)}) and (%
\ref{CAR(2;l;d)}) respectively) the $nl$-step-ahead optimal linear predictor
is given by eq. (\ref{CE-TVAR(P)}) (with $k=nl$) in Proposition (\ref%
{ProOptPred}) where%
\begin{eqnarray*}
\sum_{i=0}^{nl-1}\xi _{t,i}\phi _{0}(t-i)
&=&\sum_{s=0}^{l-1}\sum_{T=0}^{n-1}\xi _{t,Tl+s}\phi _{0}(t-s)\text{ \  \
(PAR model)}, \\
\sum_{i=0}^{nl-1}\xi _{t,i}\phi _{0}(t-i)
&=&\sum_{j=1}^{d+1}\sum_{s_{j}=l_{j-1}+1}^{l_{j}}\sum_{T=0}^{n-1}\xi
_{t,Tl+s_{j}}\phi _{0}(t-s_{j})\text{ \  \ (CAR model)},
\end{eqnarray*}%
and $\xi _{t,Tl+s}$, $\xi _{t,Tl+s_{j}}$ are given in Proposition (\ref%
{ProPAR(2;l)KSI}) and Corollary (\ref{CorCAR(2,l,r)KSIX}) respectively (see
Appendix B). \newline
Finally, for the ABAR($2;r$) model in eq. (\ref{ABAR(2)}) the $k$-step-ahead
optimal linear predictor is given by eq. (\ref{CE-TVAR(P)}) where 
\begin{equation*}
\sum_{i=0}^{k-1}\xi _{t,i}\phi _{0}(t-i)=\sum_{j=1}^{r+1}\phi
_{0,j}\sum_{i=k_{j-1}}^{kj-1}\xi _{t,i},
\end{equation*}%
and $\xi _{t,i}$ is given either in eq. (\ref{KSIABAR2}) or in Proposition (%
\ref{ProABAR(2,n)KSI}) (see Appendix B).
\end{corollary}

Franses and Paap (2005) employ the vector season representation to compute
forecasts and forecast error variances for a PAR($1;4$) process. In this way
forecasts can be generated along the same lines with quadrivariate VAR($1$)
models. Franses (1996a) derives multi-step forecast error variances for
low-order PAR models with $l=4$, using the VS representation. But, if $l$ is
large even low order specifications will have large VAR representations and
this is a handicap especially for forecasting. In contrast, our formulae
using the univariate framework allow a fast computation of the
multi-step-ahead predictors even if $l$ is large.

In what follows we give conditions for the first and second unconditional
moments of model (\ref{TVAR(P)}) to exist.\qquad

\subsection{Wold Representation}

First, we need an assumption.

Assumption A.1. $\sum_{i=0}^{k}\xi _{t,i}\phi _{0}(t-i)$ as $k\rightarrow
\infty $ converges $\forall $ $t$ and $\tsum \nolimits_{i=0}^{\infty
}\sup_{t}(\xi _{t,i}^{2}\sigma _{t-i}^{2})<M_{u}<\infty $ ($M_{u}\in \mathbb{%
R}^{+}$).

Assumption A$.$1\ is a sufficient condition for the model in eq. (\ref%
{TVAR(P)}) to admit a second-order MA($\infty $) representation. A necessary
but not sufficient condition for $\sum_{i=0}^{k}\xi _{t,i}\phi _{0}(t-i)$ to
converge is $\lim_{k\rightarrow \infty }[\xi _{t,k}\phi _{0}(t-k)]=0$ $%
\forall $ $t$. A sufficient condition for this limit to be zero is: $%
\lim_{k\rightarrow \infty }\xi _{t,k}=0$ and $\phi _{0}(t-k)$ is bounded.

Another immediate consequence of Theorem \ref{TheoGenSol}\ is the following
Proposition, where we state an expression for the first unconditional moment
of $y_{t}$.

\begin{proposition}
\label{ProWoldRepr}Let Assumption A.1 hold. Then for the TV-AR($2$) model we
have: 
\begin{equation}
y_{t}=\lim_{k\rightarrow \infty }y_{t,k}^{gen}\overset{L_{2}}{=}%
\lim_{k\rightarrow \infty }y_{t,k}^{par}\overset{L_{2}}{=}\Xi _{t,\infty
}(B)[\phi _{0}(t)+\varepsilon _{t}]=\sum_{i=0}^{\infty }\xi _{t,i}B^{i}[\phi
_{0}(t)+\varepsilon _{t}],  \label{TVAR(P)IMA}
\end{equation}%
is a unique solution of the TV-AR(2) model in eq. (\ref{TVAR(P)}). The above
expression states that $\{y_{t,k}^{par},t\in \mathbb{Z}\}$ (defined in eq. (%
\ref{TVAR(p)SOL})) $L_{2}$ converges as $k\rightarrow \infty $ if and only
if $\sum_{i=0}^{k}\xi _{t,i}\phi _{0}(t-i)$ converges and $\sum_{i=0}^{k}\xi
_{t,i}\varepsilon _{t-i}$ converges a.s., and thus under assumption A.1 $%
\lim_{k\rightarrow \infty }y_{t,k}^{gen}\overset{L_{2}}{=}\lim_{n\rightarrow
\infty }y_{t,k}^{par}$ satisfies eq. (\ref{TVAR(P)}). \newline
In other words $\lim_{k\rightarrow \infty }y_{t,k}^{gen}$ is decomposed into
a non random part and a zero mean random part. In particular, 
%
\begin{equation}
\mathbb{E}(y_{t})=\lim_{k\rightarrow \infty }\mathbb{E}(y_{t}\left \vert
\tciFourier _{\tau _{k}}\right. )=\Xi _{t,\infty }(B)\phi
_{0}(t)=\sum_{i=0}^{\infty }\xi _{t,i}B^{i}\phi _{0}(t),  \label{E-TVAR(P)}
\end{equation}%
is the non random part of $y_{t}$ and it is an infinite sum of the
periodical drifts where the time varying coefficients are expressed as
determinants of continuant matrices (the $\xi ${\footnotesize s}), while $%
\lim_{k\rightarrow \infty }\mathbb{FE}(y_{t}\left \vert \tciFourier _{\tau
_{k}}\right. )=\sum_{i=0}^{\infty }\xi _{t,i}\varepsilon _{t-i}$ is the zero
mean random part%
. Therefore the $\xi _{t,i}$\ as defined in eq. (\ref{KSIAR(2)}%
) are the Green functions associated with the second order time varying AR
polynomial: $\Phi _{t}(B)=1-\phi _{1}(t)B-\phi _{2}(t)B^{2}$.
\end{proposition}

\subsection{Second Moments}

In this subsection we state as a Proposition the result for the second
moment structure.

\begin{proposition}
\label{ProSecMom}Let Assumption A.1 hold. Then the second unconditional
moment of $y_{t}$ exists and it is given by%
\begin{equation}
\mathbb{E}(y_{t}^{2})=[\mathbb{E}(y_{t})]^{2}+\Xi _{t,\infty
}^{(2)}(B)\sigma _{t}^{2}=[\mathbb{E}(y_{t})]^{2}+\sum_{i=0}^{\infty }\xi
_{t,i}^{2}B^{i}\sigma _{t}^{2}.  \label{Var-TVAR(P)}
\end{equation}%
That is, the time varying variance of $y_{t}$ is an infinite sum of the time
varying variances of the errors with time varying coefficients (the squared
values of the $\xi ${\footnotesize s}). \newline
In addition, the time varying autocovariance function $\gamma _{t,k}$ is
given by 
\begin{eqnarray}
\gamma _{t,k} &=&\mathbb{C}ov(y_{t},y_{\tau _{k}})=\sum_{i=0}^{\infty }\xi
_{t,k+i}\xi _{\tau _{k},i}\sigma _{\tau _{k}-i}^{2}=\xi _{t,k}\mathbb{V}%
ar(y_{\tau _{k}})+  \label{Covar-TVAR(P)} \\
&&\phi _{2}(\tau _{k}+1)\xi _{t,k-1}\mathbb{C}ov(y_{\tau _{k}},y_{\tau
_{k}-1}),  \notag
\end{eqnarray}%
where the second equality follows from the MA($\infty $) representation of $%
y_{t}$ in eq. (\ref{TVAR(P)IMA}) and the third one from eq. (\ref{TVAR(p)SOL}%
) in Theorem \ref{TheoGenSol}. For any fixed $t$, $\lim_{k\rightarrow \infty
}\gamma _{t,k}=0$ when $\lim_{k\rightarrow \infty }\xi _{t,k}=0$ $\forall $ $%
t$. Finally, recall that for the PAR and ABAR models the $\xi $%
{\footnotesize s} are given either in eqs. (\ref{BlockToeplitz}) and (\ref%
{KSIABAR2}) respectively, or in Propositions (\ref{ProPAR(2;l)KSI}) and (\ref%
{ProABAR(2,n)KSI}) respectively.
\end{proposition}

Although it may be difficult to compute the covariance structure of $%
\{y_{t}\}$ explicitly, for numerical work, one can always calculate it by
computing the Green functions (that is, the continuant determinants $\xi $%
{\footnotesize s}) with eqs. (\ref{FAIMAR(p)X}) and (\ref{KSIAR(2)}) and
summing these with eq. (\ref{CE-TVAR(P)}).

\section{Conclusions}

We have provided the general solutions to low order TV-AR models in terms of
their homogeneous and particular parts. \ Our first step was to find the
fundamental set of solutions by computing the determinants of the matrix of
coefficients associated with the infinite linear system that represents the
difference equation.

The framework developed in Section 2, proved itself to be a general time
varying theory, encompassing a number of seemingly unrelated models,
discussed in Sections 3 and 4. \ We have identified common properties
(throughout the paper and in particular in Section 5), which are basic to
each of the particular application.

We believe that time varying models should take center stage in the time
series literature; this is why we have labored to develop a theory with
rigorous foundations that can encompass a variety of dynamic systems, i.e.,
periodic and cyclical processes, and AR models which contain multiple
structural breaks. \ Work that remains to be done by us and fellow
researchers is on estimation and testing (for one application on this front
see the paper by Karanasos et al., 2013) to demonstrate the usefulness of
time varying models. \ In the long run, a sound mathematical theory has to
be cointegrated with its applicability.

\bigskip

REFERENCES

\bigskip

Abdrabbo, N. A. and Priestley, M. B. (1967) On the prediction of
non-stationary processes. \textit{Journal of the Royal Statistical Society},
Series B \textbf{29}, 570-85.

Adams, G. J. and Goodwin, G. C. (1995) Parameter estimation for periodic
ARMA models. \textit{Journal of Time Series Analysis} \textbf{16}, 127-45.

Anderson, P. L., Meerschaert, M. M. and Zhang, K. (2013) Forecasting with
prediction intervals for periodic autoregressive moving average models. 
\textit{Journal of Time Series Analysis} \textbf{34}, 187-93.

Anderson, P. L. and Vecchia, A. V. (1993) Asymptotic results for periodic
autoregressive moving-average processes. \textit{Journal of Time Series
Analysis} \textbf{14}, 1-18.

Bentarzi, M. and Hallin, M. (1994) On the invertibility of periodic
moving-average models. \textit{Journal of Time Series Analysis} \textbf{15},
263-68.

Birchenhall, C. R., Bladen-Hovell, R. C., Chui, A. P. L., Osborn, D. R. and
Smith, J. P. (1989) A seasonal model of consumption. \textit{The} \textit{%
Economic Journal} \textbf{99}, 837-43.

Bollerslev, T. and Ghysels, E. (1996) Periodic autoregressive conditional
heteroscedasticity. \textit{Journal of Business \& Economic Statistics} 
\textbf{14}, 139-51.

Cipra, T. and Tlust\'{y}, P. (1987) Estimation in multiple
autoregressive-moving average models using periodicity. \textit{Journal of
Time Series Analysis} \textbf{8}, 293-300.

del Barrio Castro, T. and Osborn, D. R. (2008) Testing for seasonal unit
roots in periodic integrated autoregressive processes. \textit{Econometric
Theory} \textbf{24}, 1093--129

del Barrio Castro, T. and Osborn, D. R. (2012) Non-parametric testing for
seasonally and periodically integrated processes. \textit{Journal of Time
Series Analysis} \textbf{33}, 424-37.

Franses, P. H. (1994) A multivariate approach to modeling univariate
seasonal time series. \textit{Journal of Econometrics} \textbf{63}, 133--51.

Franses, P. H. (1996a) Multi-step forecast error variances for periodically
integrated time series. \textit{Journal of Forecasting} \textbf{15}, 83-95.

Franses, P. H. (1996b) Periodicity and Stochastic Trends in Economic Time
Series. \textit{Oxford University Press}.

Franses, P. H. and Paap, R. (2004) Periodic Time Series models. Oxford. 
\textit{Oxford University Press.}

Franses, P. H. and Paap, R. (2005) Forecasting with periodic autoregressive
time-series models, in A Companion to Economic Forecasting, Eds. Clements,
M. P. and Hendry, D. F. Oxford. \textit{Wiley}-\textit{Blackwell}, pp.
432-52.

Ghysels, E. and Osborn, D. R. (2001) The Econometric Analysis of Seasonal
Time Series. \textit{Cambridge University Press}.

Gladyshev, E. G. (1961) On periodically correlated random sequences. \textit{%
Soviet Mathematics }\textbf{2}, 385-88.

Granger, C. W. J. (2007) Forecasting - looking back and forward: Paper to
celebrate the 50th anniversary of the econometrics institute at the Erasmus
University, Rotterdam. \textit{Journal of Econometrics} \textbf{138}, 3-13.

Granger, C. W. J. (2008) Non-linear models: Where do we go next - time
varying parameter models? \textit{Studies in Nonlinear Dynamics and
Econometrics} \textbf{12}, 1-9.

Hurd, H. L. and Miamee, A. (2007) Periodically Correlated Random Sequences:
Spectral Theory and Practice. Hoboken, New Jersey. \textit{Wiley-Blackwell}.

Karanasos, M. (2001) Prediction in ARMA models with GARCH in mean effects. 
\textit{Journal of Time Series Analysis} \textbf{22}, 555-76.

Karanasos, M., Paraskevopoulos, A. G. and Dafnos, S. (2013) A univariate
time varying analysis of periodic ARMA processes. \textit{Unpublished Paper}.

Karanasos, M., Paraskevopoulos, A. G., Menla Ali, F., Karoglou, M. and
Yfanti, S. (2013) Modelling returns and volatilities during financial
crises: a time varying coefficient approach. \textit{Unpublished Paper}.

Lund, R. and Basawa, I. V. (2000) Recursive prediction and likelihood
evaluation for periodic ARMA models. \textit{Journal of Time Series Analysis}
\textbf{21}, 75-93.

Lund, R., Shao, Q. and Basawa, I. (2006) Parsimonious periodic time series
modeling. \  \textit{Australian \& New Zealand Journal of Statistics} \textbf{%
48}, 33-47.

McLeod, A. I. (1994) Diagnostic checking of periodic autoregression models
with application. \textit{Journal of Time Series Analysis} \textbf{15},
221-33.

Osborn, D. R. (1988) Seasonality and habit persistence in a life cycle model
of consumption. \textit{Journal of Applied Econometrics }\textbf{3}, 255--66.

Osborn, D. R. (1991) The implications of periodically varying coefficients
for seasonal time-series processes. \textit{Journal of Econometrics} \textbf{%
48}, 373--84.

Osborn, D. R. and Smith, J. P. (1989) The performance of periodic
autoregressive models in forecasting seasonal U.K. consumption. \textit{%
Journal of Business \& Economic Statistics} \textbf{7}, 117-27.

Pagano, M. (1978) On periodic and multiple autoregressions. \  \textit{The} 
\textit{Annals of Statistics} \textbf{6}, 1310-317.

Paraskevopoulos, A. G. (2012) The Infinite Gauss-Jordan elimination on
row-finite $\omega \times \omega $ matrices. arXiv: 1201.2950.

Paraskevopoulos, A. G. and Karanasos, M. (2013) Closed form solutions for
linear difference equations with time dependent coefficients. \textit{%
Unpublished paper}.

Paraskevopoulos, A. G., Karanasos, M. and Dafnos, S. (2013) A unified theory
for time varying models: foundations with applications in the presence of
breaks and heteroskedasticity (and some results on companion and Hessenberg
matrices). \textit{Unpublished Paper}.

Pesaran, M. H., Pettenuzzo, D. and Timmermann, A. (2006) Forecasting time
series subject to multiple structural breaks. \textit{Review of Economic
Studies} 73, 1057-084.

Rao, S. T. (1970) The fitting of non-stationary time-series models with
time-dependent parameters. \textit{Journal of the Royal Statistical Society,}
Series B \textbf{32}, 312-22.

Shao, Q. (2008) Robust estimation for periodic autoregressive time series. 
\textit{Journal of Time Series Analysis} \textbf{29}, 251-63.

Sydsaeter, K., Hammond, P. J., Seierstad, A. and Strom, A. (2008) Further
Mathematics for Economic Analysis. New York. \textit{Prentice Hall, }2nd
edition.

Taylor, A. M. R. (2002) Regression-based unit root tests with recursive mean
adjustment for seasonal and nonseasonal time series. \textit{Journal of
Business \& Economic Statistics} \textbf{20}, 269--81.

Taylor, A. M. R. (2003) On the asymptotic properties of some seasonal unit
root tests. \textit{Econometric Theory} \textbf{19}, 311--21.

Taylor, A. M. R. (2005) Variance ratio tests of the seasonal unit root
hypothesis. \textit{Journal of Econometrics} \textbf{124}, 33--54.

Tesfaye Y. G., Anderson, P. L. and Meerschaert, M. M. (2011) Asymptotic
results for Fourier-PARMA time series. \textit{Journal of Time Series
Analysis} \textbf{32}, 157--74.

Tiao, G. C. and Grupe, M. R. (1980) Hidden periodic autoregressive-moving
average models in time series data. \textit{Biometrika} \textbf{67}, 365--73.

Tiao, G. C. and Guttman, I. (1980) Forecasting contemporal aggregates of
multiple time series. \textit{Journal of Econometrics} \textbf{12}, 219--30.

Whittle, P. (1965) Recursive relations for predictors of non-stationary
processes. \textit{Journal of the Royal Statistical Society,} Series B 
\textbf{27}, 523-32.

\appendix

\section{APPENDIX}

In this appendix we prove Theorem 1. Before proceeding with the main body of
the proof, we present two essential tools for carrying it out.

\textbf{The Infinite Gaussian Elimination}. Following Paraskevopoulos
(2012), we apply the infinite Gaussian elimination algorithm implemented
under a rightmost pivot strategy to the coefficient matrix $%
\mbox{\boldmath$\Phi$}$ of (\ref{PHIMatrix}). The process is briefly
described below. \newline
Call $\mathbf{h}^{(1)}=\mathbf{H}^{(1)}=(-\phi _{2}(\tau _{k}+1),-\phi
_{1}(\tau _{k}+1),1,0,...)$ the opposite-sign first row of $\mathbf{\Phi }$.
Insert the second row of $\mathbf{\Phi }$ below $\mathbf{H}^{(1)}$ to build
the matrix $\mathbf{B}^{(2)}$:%
\begin{equation*}
\mathbf{B}^{(2)}=\left( \! \!%
\begin{array}{ccccc}
-\phi _{2}(\tau _{k}+1) & -\phi _{1}(\tau _{k}+1) & 1 & 0 & \ldots \\ 
0 & \phi _{2}(\tau _{k}+2) & \phi _{1}(\tau _{k}+2) & -1 & \ldots%
\end{array}%
\right) .
\end{equation*}%
Use as pivot the rightmost one of $\mathbf{H}^{(1)}$ to clear the element $%
\phi _{1}(\tau _{k}+2)$ in the second row of $\mathbf{B}^{(2)}$. After
normalization it yields the matrix: 
\begin{equation*}
\mathbf{H}^{(2)}=\left( \! \!%
\begin{array}{ccccc}
-\phi _{2}(\tau _{k}+1) & -\phi _{1}(\tau _{k}+1)\! \! & 1 & 0 & \ldots \\ 
-\phi _{2}(\tau _{k}+1)\phi _{1}(\tau _{k}+2) & -\phi _{2}(\tau _{k}+2)-\phi
_{1}(\tau _{k}+1)\phi _{1}(\tau _{k}+2) & 0 & 1 & \ldots%
\end{array}%
\right) .
\end{equation*}%
{\normalsize Insert the third row of }$\mathbf{\Phi }${\normalsize \ below $%
\mathbf{H}^{(2)}$ to build the matrix $\mathbf{B}^{(3)}$: 
\begin{equation*}
\left( \! \! \! \!%
\begin{array}{ccccrr}
-\phi _{2}(\tau _{k}\!+\!1)\! \! \! & -\phi _{1}(\tau _{k}\!+\!1)\! \! \! \!
\! & 1 & 0 & 0 & \ldots \\ 
-\phi _{2}(\tau _{k}\!+\!1)\phi _{1}(\tau _{k}\!+\!2)\! \! \! & -\phi
_{2}(\tau _{k}\!+\!2)\!-\! \phi _{1}(\tau _{k}\!+\!1)\phi _{1}(\tau
_{k}\!+\!2) & 0 & 1 & 0 & \ldots \\ 
0\! \! \! & 0 & \phi _{2}(\tau _{k}+3) & \phi _{1}(\tau _{k}+3) & -1 & \ldots%
\end{array}%
\right) .
\end{equation*}%
Use the first two rows of $\mathbf{B}^{(3)}$ as pivot rows and their
rightmost }${\normalsize 1}${\normalsize s as pivot elements to clear the
entries $\phi _{2}(\tau _{k}+3)$ and $\phi _{1}(\tau _{k}+3)$ of $\mathbf{B}%
^{(3)}$, producing the matrix $\mathbf{H}^{(3)}$: 
\begin{equation*}
\mathbf{H}^{(3)}=\left( 
\begin{array}{ccccccc}
h_{11} & h_{12} & 1 & 0 & 0 & 0 & ... \\ 
h_{21} & h_{22} & 0 & 1 & 0 & 0 & ... \\ 
h_{31} & h_{32} & 0 & 0 & 1 & 0 & ...\vspace{-0.05in}%
\end{array}%
\right) .
\end{equation*}%
where the entries of the first column of }$\mathbf{H}^{(3)}$ {\normalsize %
are given by 
\begin{equation*}
\begin{array}{l}
h_{11}=-\phi _{2}(\tau _{k}+1),\ h_{21}=-\phi _{2}(\tau _{k}+1)\phi
_{1}(\tau _{k}+2), \\ 
h_{31}=-\phi _{2}(\tau _{k}+1)\phi _{1}(\tau _{k}+2)\phi _{1}(\tau
_{k}+3)\!-\! \phi _{2}(\tau _{k}+1)\phi _{2}(\tau _{k}+3),...%
\end{array}%
\end{equation*}%
and the entries of the second column are given by 
\begin{equation*}
\begin{array}{l}
h_{12}=-\phi _{1}(\tau _{k}+1),\ h_{22}=-\phi _{2}(\tau _{k}+2)\!-\! \phi
_{1}(\tau _{k}+1)\phi _{1}(\tau _{k}+2), \\ 
h_{32}=-\phi _{1}(\tau _{k}\! \!+\!1)\phi _{1}(\tau _{k}\! \!+\!2)\phi
_{1}(\tau _{k}\! \!+\!3)\!-\! \phi _{2}(\tau _{k}\! \!+\!2)\phi _{1}(\tau
_{k}\! \!+\!3)\!-\! \phi _{2}(\tau _{k}\! \!+\!3)\phi _{1}(\tau _{k}\!
\!+\!1)%
\end{array}%
.
\end{equation*}%
This process continues ad infinitum, generating an infinite chain of
submatrices}%
\begin{equation*}
{\normalsize \mathbf{H}^{(1)}\sqsubset \mathbf{H}^{(2)}\sqsubset \mathbf{H}%
^{(3)}}\mathbf{\sqsubset \ldots \sqsubset }{\normalsize \mathbf{H}}
\end{equation*}%
whose limit row-finite matrix ${\normalsize \mathbf{H}}$ is {\normalsize the
Hermite Form (HF) of }$\mathbf{\Phi }$. The $i$th row of ${\normalsize 
\mathbf{H}}$ is defined to be the last row of ${\normalsize \mathbf{H}^{(i)}}
$.

\textbf{Two Fundamental Solutions}. {\normalsize The opposite-sign two first
columns of $\mathbf{H}$ augmented at the top by $(1,0)$ and $(0,1)$,
respectively, that is 
\begin{equation*}
\begin{array}{ll}
\mathbf{\xi }_{\tau _{k}}^{(2)}=\! \! \! \! & (1,\  \ 0,\  \  \phi _{2}(\tau
_{k}+1),\  \  \phi _{2}(\tau _{k}+1)\phi _{1}(\tau _{k}+2), \\ 
\! \! \! \! & \phi _{2}(\tau _{k}+1)\phi _{1}(\tau _{k}+2)\phi _{1}(\tau
_{k}+3)\!+\! \phi _{2}(\tau _{k}+1)\phi _{2}(\tau _{k}+3),...)^{\prime }, \\ 
\mathbf{\xi }_{\tau _{k}}^{(1)}=\! \! \! \! & (0,\  \ 1,\  \  \phi _{1}(\tau
_{k}+1),\  \  \phi _{2}(\tau _{k}+2)\!+\! \phi _{1}(\tau _{k}+1)\phi _{1}(\tau
_{k}+2), \\ 
\! \! \! \! & \phi _{1}(\tau _{k}\! \!+\!1)\phi _{1}(\tau _{k}\! \!+\!2)\phi
_{1}(\tau _{k}\! \!+\!3)\!+\! \phi _{2}(\tau _{k}\! \!+\!2)\phi _{1}(\tau
_{k}\! \!+\!3)\!+\! \phi _{2}(\tau _{k}\! \!+\!3)\phi _{1}(\tau _{k}\!
\!+\!1),...)^{\prime }%
\end{array}%
\end{equation*}%
are the two linearly independent solution sequences of the space of
homogeneous solutions of eq. (\ref{Difference(P)}). The linear independence
of $\mathbf{\xi }_{\tau _{k}}^{(1)},\mathbf{\xi }_{\tau _{k}}^{(2)}$ follows
from the fact that they possess the Casoratian: 
\begin{equation*}
\det \left( 
\begin{array}{ll}
1 & 0 \\ 
0 & 1%
\end{array}%
\right) \not=0.
\end{equation*}%
}

{\normalsize We observe that the terms of the sequences }$\mathbf{\xi }$%
{\normalsize $_{\tau _{k}}^{(1)},\mathbf{\xi }_{\tau _{k}}^{(2)}$ are
expansions of the following determinants 
\begin{equation}
\begin{array}{ll}
\mathbf{\xi }_{\tau _{k}}^{(2)}\! \! \! \! & =\left \{ 
\begin{array}{ll}
1 &  \\ 
0 &  \\ 
\phi _{2}(\tau _{k}+1) &  \\ 
\det \left( 
\begin{array}{cc}
\phi _{2}(\tau _{k}+1) & -1 \\ 
0 & \phi _{1}(\tau _{k}+2)%
\end{array}%
\right) & \vspace{0.1in} \\ 
\det \left( 
\begin{array}{ccc}
\phi _{2}(\tau _{k}+1) & -1 & 0 \\ 
0 & \phi _{1}(\tau _{k}+2) & -1 \\ 
0 & \phi _{2}(\tau _{k}+3) & \phi _{1}(\tau _{k}+3)%
\end{array}%
\! \! \right) ,\vspace{-0.05in} &  \\ 
.\vspace{-0.1in} &  \\ 
.\vspace{-0.1in} &  \\ 
. & 
\end{array}%
\right.%
\end{array}
\tag{A.1}
\end{equation}%
\begin{equation}
\begin{array}{ll}
\mathbf{\xi }_{\tau _{k}}^{(1)}\! \! \! \! & =\left \{ 
\begin{array}{ll}
0 &  \\ 
1 &  \\ 
\phi _{1}(\tau _{k}+1) &  \\ 
\det \left( 
\begin{array}{cc}
\phi _{1}(\tau _{k}+1) & -1 \\ 
\phi _{2}(\tau _{k}+2) & \phi _{1}(\tau _{k}+2)%
\end{array}%
\right) & \vspace{0.1in} \\ 
\det \left( 
\begin{array}{ccc}
\phi _{1}(\tau _{k}+1) & -1 & 0 \\ 
\phi _{2}(\tau _{k}+2) & \phi _{1}(\tau _{k}+2) & -1 \\ 
0 & \phi _{2}(\tau _{k}+3) & \phi _{1}(\tau _{k}+3)%
\end{array}%
\! \! \right) .\vspace{-0.05in} &  \\ 
.\vspace{-0.1in} &  \\ 
.\vspace{-0.1in} &  \\ 
. & 
\end{array}%
\right.%
\end{array}
\tag{A.2}
\end{equation}%
The first few terms of the homogeneous solution sequences, as shown above,
suggest that the general terms of }$\mathbf{\xi }${\normalsize $_{\tau
_{k}}^{(1)},\mathbf{\xi }_{\tau _{k}}^{(2)}$ are 
\begin{equation}
\xi _{t,k}^{(m)}=\det (\mathbf{\Phi }_{t,k}^{(m)}),\  \ m=1,2,  \tag{A.3}
\end{equation}%
where $\mathbf{\Phi }_{t,k}^{(1)}=\mathbf{\Phi }_{t,k}$ and} $\xi
_{t,k}^{(1)}=\xi _{t,k}$ (we drop the superscript $1$ for notational
convenience){\normalsize , as introduced in eqs. (\ref{FAIMAR(p)X}) and (\ref%
{KSIAR(2)}), and 
\begin{equation*}
\mathbf{\Phi }_{t,k}^{(2)}=\left( 
\begin{array}{cccccc}
\phi _{2}(\tau _{k}+1) & -1 &  &  &  &  \\ 
& \phi _{1}(\tau _{k}+2) & -1 &  &  &  \\ 
& \phi _{2}(\tau _{k}+3) & \phi _{1}(\tau _{k}+3) & -1 &  &  \\ 
&  & \ddots & \ddots & \ddots &  \\ 
&  &  & \phi _{2}(t-1) & \phi _{1}(t-1) & -1 \\ 
&  &  &  & \phi _{2}(t) & \phi _{1}(t)%
\end{array}%
\right) .
\end{equation*}%
In the following Proposition we use mathematical induction to verify the
above generalization formally. }

\begin{proposition}
{\normalsize \label{PropAlsecondOrder} The general terms of the fundamental
solution sequences }$\mathbf{\xi }${\normalsize $_{\tau _{k}}^{(m)}$, $m=1,2$%
, are given by eq. \emph{(A.3)}, that is 
\begin{equation}
\xi _{t,k}^{(2)}=\det \left( 
\begin{array}{cccccc}
\phi _{2}(\tau _{k}+1) & -1 &  &  &  &  \\ 
& \phi _{1}(\tau _{k}+2) & -1 &  &  &  \\ 
& \phi _{2}(\tau _{k}+3) & \phi _{1}(\tau _{k}+3) & -1 &  &  \\ 
&  & \ddots & \ddots & \ddots &  \\ 
&  &  & \phi _{2}(t-1) & \phi _{1}(t-1) & -1 \\ 
&  &  &  & \phi _{2}(t) & \phi _{1}(t)%
\end{array}%
\right)  \tag{A.4}
\end{equation}%
and 
\begin{equation}
\xi _{t,k}=\det \left( 
\begin{array}{cccccc}
\phi _{1}(\tau _{k}+1) & -1 &  &  &  &  \\ 
\phi _{2}(\tau _{k}+2) & \phi _{1}(\tau _{k}+2) & -1 &  &  &  \\ 
& \phi _{2}(\tau _{k}+3) & \phi _{1}(\tau _{k}+3) & -1 &  &  \\ 
&  & \ddots & \ddots & \ddots &  \\ 
&  &  & \phi _{2}(t-1) & \phi _{1}(t-1) & -1 \\ 
&  &  &  & \phi _{2}(t) & \phi _{1}(t)%
\end{array}%
\right) .  \tag{A.5}
\end{equation}%
}
\end{proposition}

\begin{proof}
{\normalsize If $t=\tau _{k}+1$ and $t=\tau _{k}+2$ then $\xi _{\tau
_{k}+1,1}$ and $\xi _{\tau _{k}+2,2}$ is the third term and fourth term of
the sequences as directly verified by eq. (A.}$2${\normalsize ). We assume
that $\xi _{t-2,k-2}$ and $\xi _{t-1,k-1}$ are terms of $\mathbf{\xi }_{\tau
_{k}}^{(1)}$. We show that $\xi _{t,k}$ is also a term of $\mathbf{\xi }%
_{\tau _{k}}^{(1)}$. Expanding $\xi _{t,k}$ along the last row and taking
into account that $\mathbf{\Phi }_{t,k}$ is a $k\times k$ matrix, we have: 
\begin{equation*}
\begin{array}{c}
\xi _{t,k}=(-1)^{2k}\phi _{1}(t)\det \left( \! \!%
\begin{array}{cccccc}
\phi _{1}(\tau _{k}+1) & -1 &  &  &  &  \\ 
\phi _{2}(\tau _{k}+2) & \phi _{1}(\tau _{k}+2) & -1 &  &  &  \\ 
& \phi _{2}(\tau _{k}+3) & \phi _{1}(\tau _{k}+3) & -1 &  &  \\ 
&  & \ddots & \ddots & \ddots &  \\ 
&  &  & \phi _{2}(t-2) & \phi _{1}(t-2) & -1 \\ 
&  &  &  & \phi _{2}(t-1) & \phi _{1}(t-1)%
\end{array}%
\right) +\vspace{0.1in} \\ 
(-1)^{2k-1}(-1)\phi _{2}(t)\det \left( \! \! \!%
\begin{array}{cccccc}
\phi _{1}(\tau _{k}+1) & -1 &  &  &  &  \\ 
\phi _{2}(\tau _{k}+2) & \phi _{1}(\tau _{k}+2) & -1 &  &  &  \\ 
& \phi _{2}(\tau _{k}+3) & \phi _{1}(\tau _{k}+3) & -1 &  &  \\ 
&  & \ddots & \ddots & \ddots &  \\ 
&  &  & \phi _{2}(t-3) & \phi _{1}(t-3) & -1 \\ 
&  &  &  & \phi _{2}(t-2) & \phi _{1}(t-2)%
\end{array}%
\! \right) .%
\end{array}%
\end{equation*}%
Using the induction hypothesis, the above result can be written as 
\begin{equation*}
\xi _{t,k}=\phi _{1}(t)\xi _{t-1,k-1}+\phi _{2}(t)\xi _{t-2,k-2},
\end{equation*}%
which shows that $\xi _{t,k}$ is a homogeneous solution of (\ref%
{Difference(P)}). Thus $\xi _{t,k}$ in (A.5) is a term of the solution
sequence and the induction is complete. By analogy, we can show (A.4) and
the proof is complete. }
\end{proof}

{\normalsize The fundamental solution }$\xi _{t,k}$ (respectively $\xi
_{t,k}^{(2)}$) {\normalsize can be obtained by augmenting the core solution
matrix $\mathbf{C}_{t,k}$ (see eq. (\ref{CoreMatrix}) in the main body of
the paper) on the left by a $k\times 1$ column consisting of the first $k$
entries of the second column (respectively of the first column) of }$\mathbf{%
P}$ or $\mathbf{\Phi }${\normalsize . }

\begin{proof}
(\textbf{of Theorem 1}) {\normalsize As a direct consequence of Proposition
1, the general homogeneous solution of eq. (\ref{Difference(P)}) is the
linear combination of the fundamental solutions as given below: 
\begin{equation}
y_{t,k}^{hom}=\xi _{t,k}y_{\tau _{k}}+\xi _{t,k}^{(2)}y_{\tau _{k}-1}. 
\tag{A.6}
\end{equation}%
By expanding $\xi _{t,k}^{(2)}$ along the first column we obtain 
\begin{equation*}
\xi _{t,k}^{(2)}=\phi _{2}(\tau _{k}+1)\xi _{t,k-1}
\end{equation*}%
and therefore (A.6) takes the form 
\begin{equation*}
y_{t,k}^{hom}=\xi _{t,k}y_{\tau _{k}}+\phi _{2}(\tau _{k}+1)\xi
_{t,k-1}y_{\tau _{k}-1},
\end{equation*}%
which coincides with the general homogeneous solution employed in eq. (7). }%
\newline
{\normalsize Next we show that $y_{t,k}^{par}$, employed in eq. (\ref%
{TVAR(p)SOL}), is a particular solution of eq. (\ref{Difference(P)}). Using
the same arguments as in the proof of Proposition \ref{PropAlsecondOrder} we
can show that 
\begin{equation}
y_{t,k}^{par}=\det \left( 
\begin{array}{cccccc}
\phi _{0}(\tau _{k}+1)+\epsilon _{\tau _{k}+1} & -1 &  &  &  &  \\ 
\phi _{0}(\tau _{k}+2)+\epsilon _{\tau _{k}+2} & \phi _{1}(\tau _{k}+2) & -1
&  &  &  \\ 
\phi _{0}(\tau _{k}+3)+\epsilon _{\tau _{k}+3} & \phi _{2}(\tau _{k}+3) & 
\phi _{1}(\tau _{k}+3) & -1 &  &  \\ 
\vdots &  & \ddots & \ddots & \ddots &  \\ 
\phi _{0}(t-1)+\epsilon _{t-1} &  &  & \phi _{2}(t-1) & \phi _{1}(t-1) & -1
\\ 
\phi _{0}(t)+\epsilon _{t} &  &  &  & \phi _{2}(t) & \phi _{1}(t)%
\end{array}%
\right) ,  \tag{A.7}
\end{equation}%
is the solution of the initial value problem determined by eq. (\ref%
{Difference(P)}) subject to the initial values $y_{-1}=y_{0}=0$. This is the
determinant of the core solution matrix $\mathrm{\mathbf{C}}_{t,k}$
augmented on the left by a $k\times 1$ column consisting of the opposite
sign first $k$ entries of the right-hand side sequence of eq. (\ref%
{Difference(P)}). }\newline
{\normalsize Now expanding the determinant in eq. (A.7) along the first
column we obtain $y_{t,k}^{par}$ in terms of $\xi _{t,i}$ and $\phi
_{0}(t-i),\epsilon _{t-i}$ for $i=0,1,...,k-1$, as used in eq. (\ref%
{TVAR(p)SOL}). Therefore the general solution in eq. (\ref{TVAR(p)SOL}), as
the sum of the general homogeneous solution plus a particular solution, has
been established. This completes the proof of Theorem 1. }
\end{proof}

\section{APPENDIX\label{SecAppendABAR}}

In this Appendix we will make use of the block Toeplitz matrix in eq. (\ref%
{Block Toeplitz1}) to obtain an explicit formula of $\xi _{t,nl}$ in which
we decompose it into tridiagonal determinants, $\xi _{t,l}$. To prepare the
reader, before we present the main result we consider the case where $n=2$,
that is we go from time $t$ back to time $t-2l$. The tridiagonal determinant 
$\xi _{t,2l}$\  \ can be written as the sum of two terms 
\begin{eqnarray}
\xi _{t,2l} &=&\left \vert 
\begin{array}{ll}
\mathbf{\Phi }_{t,l} & \overline{\mathbf{0}} \\ 
\widetilde{\mathbf{0}}_{t} & \mathbf{\Phi }_{t,l}%
\end{array}%
\right \vert =  \TCItag{B.1} \\
&=&\xi _{t,l}^{2}+\phi _{2}(\tau _{l}+1)\xi _{t,l-1}\xi _{t-1,l-1},  \notag
\end{eqnarray}%
where each term is the product of two continuant (or tridiagonal)
determinants.

Next let $i_{j}\in \{0,1\}$, $j=1,\ldots ,n-1$, and define%
\begin{equation}
\varphi _{2,j}=\left \{ 
\begin{array}{ccc}
1 & \text{if} & i_{j}=0, \\ 
\phi _{2}(\tau _{jl}+1) & \text{if} & i_{j}=1.%
\end{array}%
\right.  \tag{B.2}
\end{equation}

\begin{proposition}
\label{ProPAR(2;l)KSI}For the PAR($2;l$) process in eq. (\ref{PAR(2)}), $\xi
_{t,nl}$ is the determinant of $\mathbf{\Phi }_{t,nl}$ in eq. (\ref%
{BlockToeplitz}), and therefore can be written as%
\begin{equation}
\xi _{t,nl}=\dsum \limits_{i_{1}=0}^{1}\cdots \dsum
\limits_{i_{n-1}=0}^{1}\{ \xi _{t,l-i_{1}}(\dprod \limits_{T=2}^{n-1}\varphi
_{2,T-1}\xi _{t-i_{T-1},l-i_{T}-i_{T-1}})\varphi _{2,n-1}\xi
_{t-i_{n-1},l-i_{n-1}}\},  \tag{B.3}
\end{equation}%
where $\dsum \cdots \dsum $ stands for a multiple but finite summation, and
recall that $\xi _{t,l}=\left \vert \mathbf{\Phi }_{t,l}\right \vert $ and $%
\mathbf{\Phi }_{t,l}$ is given by eq. (\ref{FAIM1}).
\end{proposition}

In the above Proposition $\xi _{t,nl}$ is expressed as the sum of $\dsum
\limits_{j=0}^{n-1}\binom{n-1}{j}=2^{n-1}$\ terms, each of which is the
product of $n$\ terms. In other words, it is decomposed into determinants of 
$(l-m)\times (l-m)$ continuant matrices, $m=0,1,2$: $\Phi
_{t-i_{T-1},l-i_{T}-i_{T-1}}$.

When $n=3$, eq. (B.$3$) reduces to:%
\begin{eqnarray*}
\xi _{t,nl} &=&\xi _{t,l}^{3}+\phi _{2}(\tau _{l}+1)\xi _{t,l-1}\xi
_{t-1,l-1}\xi _{t,l} \\
&&+\xi _{t,l}\phi _{2}(\tau _{l}+1)\xi _{t,l-1}\xi _{t-1,l-1} \\
&&+\phi _{2}^{2}(\tau _{l}+1)\xi _{t,l-1}\xi _{t-1,l-2}\xi _{t-1,l-1} \\
&=&\xi _{t,l}^{3}+2\phi _{2}(\tau _{l}+1)\xi _{t,l-1}\xi _{t-1,l-1}\xi
_{t,l}+\phi _{2}^{2}(\tau _{l}+1)\xi _{t,l-1}\xi _{t-1,l-2}\xi _{t-1,l-1},
\end{eqnarray*}%
that is, $\xi _{t,nl}$ is equal to the sum of four ($p^{n-1}=2^{2}$; $%
i_{1}=i_{2}=0$, $i_{1}=i_{2}=1$, $i_{1}=0$ and $i_{2}=1$, and $i_{1}=1$ and $%
i_{2}=0$) terms, each of which is the product of three ($n=3$) $\xi $'s
(continuant determinants).

Next we will prove Proposition \ref{ProPAR(2;l)KSI} by mathematical
induction. For $n=2$ the result has been proved in eq. (B.$1$). If we assume
that eq. (B.$3$) holds for $n$ then it will be sufficient to prove that it
holds for $n+1$ as well.

\begin{proof}
(\textbf{Proposition} \ref{ProPAR(2;l)KSI}) Assume that 
\begin{equation}
\xi _{t,nl}=\left \vert \Phi _{t,nl}\right \vert =\dsum
\limits_{i_{1}=0}^{1}\cdots \dsum \limits_{i_{n-1}=0}^{1}\{ \xi
_{t,l-i_{1}}(\dprod \limits_{l=2}^{n-1}\varphi _{2,T-1}\xi
_{t-i_{T-1},l-i_{T}-i_{T-1}})\varphi _{2,n-1}\xi _{t-i_{n-1},l-i_{n-1}}\}. 
\tag{B.4}
\end{equation}%
Similarly to eq. (B.$1$) we can express $\xi _{t,(n+1)l}$ as the determinant
of a $2\times 2$ block matrix:%
\begin{eqnarray}
\xi _{t,(n+1)l} &=&\left \vert 
\begin{array}{ll}
\mathbf{\Phi }_{t,l} & \overline{\mathbf{0}} \\ 
\widetilde{\mathbf{0}}_{t} & \mathbf{\Phi }_{t,nl}%
\end{array}%
\right \vert =\left \vert \mathbf{\Phi }_{t,nl}\right \vert \left \vert 
\mathbf{\Phi }_{t,l}\right \vert +\phi _{2}(t-nl+1)\left \vert \mathbf{\Phi }%
_{t,nl-1}\right \vert \left \vert \mathbf{\Phi }_{t-1,l-1}\right \vert 
\notag \\
&=&\xi _{t,nl}\xi _{t,l}+\phi _{2}(t-nl+1)\xi _{t,nl-1}\xi _{t-1,l-1}, 
\TCItag{B.5}
\end{eqnarray}%
where $\widetilde{\mathbf{0}}_{t}$, is an $nl\times l$\ matrix of zeros
except for $\phi _{2}(t-nl+1)$\ in its $1\times l$\ entry and the second
equality follows from eq. (B.$1$). Combining eqs. (B.$4$) and (B.$5$) yields%
\begin{eqnarray}
&&\xi _{t,(n+1)l}%
\begin{tabular}{l}
=%
\end{tabular}%
\dsum \limits_{i_{1}=0}^{1}\cdots \dsum \limits_{i_{n-1}=0}^{1}\{ \xi
_{t,l-i_{1}}(\dprod \limits_{T=2}^{n-1}\varphi _{2,T-1}\xi
_{t-i_{T-1},l-i_{T}-i_{T-1}})\varphi _{2,n-1}\xi _{t-i_{n-1},l-i_{n-1}}\}
\xi _{t,l}+  \notag \\
&&\dsum \limits_{i_{1}=0}^{1}\cdots \dsum \limits_{i_{n-1}=0}^{1}\{ \xi
_{t,l-i_{1}}(\dprod \limits_{T=2}^{n-1}\varphi _{2,T-1}\xi
_{t-i_{T-1},l-i_{T}-i_{T-1}})\varphi _{2,n-1}\xi _{t-i_{n-1},l-1-i_{n-1}}\}
\phi _{2}(t-nl+1)\xi _{t-1,l-1}  \notag \\
&=&\dsum \limits_{i_{1}=0}^{1}\cdots \dsum \limits_{i_{n}=0}^{1}\{ \xi
_{t,l-i_{1}}(\dprod \limits_{T=2}^{n}\varphi _{2,T-1}\xi
_{t-i_{T-1},l-i_{T}-i_{T-1}})\varphi _{2,n}\xi _{t-i_{n},l-i_{n}}\}, 
\TCItag{B.6}
\end{eqnarray}%
which completes the proof.
\end{proof}

\begin{corollary}
\label{CorCAR(2,l,r)KSIX}For the CAR($2;l;d$) process, in eq. (\ref%
{CAR(2;l;d)}), with $0\leq d\leq l-1$, $\xi _{t,l}=\left \vert \mathbf{\Phi }%
_{t,l}\right \vert $ (see eq. (\ref{PHI(l)CAR2})) can be written as%
\begin{equation}
\xi _{t,l}=\dsum \limits_{i_{1}=0}^{1}\cdots \dsum \limits_{i_{d}=0}^{1}\{
\xi _{t,l_{1}-i_{1}}(\dprod \limits_{j=2}^{d}\varphi _{2,j-1}\xi
_{t-l_{j-1}-i_{j-1},l_{j}-l_{j-1}-i_{j}-i_{j-1}})\varphi _{2,d}\xi
_{t-l_{d}-i_{d},l-l_{d}-i_{d}}\},  \tag{B.7}
\end{equation}%
where $\varphi _{2,j}$ is defined similarly to the one in eq. (B.$2$), i.e., 
$\varphi _{2,j}=\phi _{2}(t-(l_{j}-l_{j-1})+1)$ if $i_{j}=1$ (the proof is
along the lines of that of Proposition (\ref{ProPAR(2;l)KSI}) above).
\end{corollary}

\begin{proposition}
\label{ProABAR(2,n)KSI}For the ABAR($2;r$) process in eq. (\ref{ABAR(2)})
with $r$, $0\leq r\leq k-1$, abrupt breaks at times $t-k_{1}$, $t-k_{2}$, $%
\ldots $, $t-k_{r}$, $\xi _{t,k}$ in eq. (\ref{KSIABAR2}) can be written as%
\begin{equation}
\xi _{t,k}=\dsum \limits_{i_{1}=0}^{1}\cdots \dsum \limits_{i_{r}=0}^{1}\{
\xi _{t,k_{1}-i_{1}}(\dprod \limits_{j=2}^{r}\varphi _{2,j-1}\xi
_{t-k_{j-1}-i_{j-1},k_{j}-k_{j-1}-i_{j}-i_{j-1}})\varphi _{2,r}\xi
_{t-k_{r}-i_{r},k-k_{r}-i_{r}}\}.  \tag{B.8}
\end{equation}%
where $\varphi _{2,j}$ is defined similarly to the one in eq. (B.$2$) (the
proof is similar to that of Proposition (\ref{ProPAR(2;l)KSI}) above).
\end{proposition}

\section{APPENDIX\label{SubSecVSRepr}}

\emph{Vector Seasons Representation}

For the benefit of the reader this Appendix reviews some results on PARMA
models. Recall that the drift and the autoregressive coefficients are
periodically varying: $\phi _{m}(t)=\phi _{m}(\tau _{n})$, $m=0,1,2$, $\tau
_{n}=\tau -nl$. Recall also that $t_{s}$ denotes time at the $s$th season: $%
t_{s}=Tl+s$, $s=1,\ldots ,l$, and that we can write $\phi _{m}(Tl+s)=\phi
_{m,s}$ (see eq. (\ref{PAR(2)})).

We assume without loss of generality that time $t$ is at the $l$th season
(e.g., $t=t_{l}=(T+1)l$). Thus our $\mathbf{\Phi }_{\mathfrak{t},l}$ matrix
in eq. (\ref{FAIM1}) becomes:

\begin{equation*}
\mathbf{\Phi }_{t,l}=\mathbf{\Phi }(l)=\left( 
\begin{array}{ccccc}
\phi _{1,1} & -1 &  &  &  \\ 
\phi _{2,2} & \phi _{1,2} & -1 &  &  \\ 
& \ddots & \ddots & \ddots &  \\ 
&  & \phi _{2,l-1} & \phi _{1,l-1} & -1 \\ 
&  &  & \phi _{2,l} & \phi _{1,l}%
\end{array}%
\right) .
\end{equation*}%
A convenient representation of the PAR($2;l$) model (\ref{PAR(2)}) is the
VAR($1$) representation- hereafter we will refer to it as the vector of
seasons (VS) representation (see, for example, Tiao and Guttman, 1980;
Osborn, 1991; \ Franses, 1994, 1996a,b; del Barrio Castro and Osborn, 2008).

The corresponding VS representation of the PAR($2;l$) model (ignoring the
drifts) is given by%
\begin{equation}
\mathbf{\Phi }_{0}\mathbf{y}_{T}=\mathbf{\Phi }_{1}\mathbf{y}_{T-1}+\mathbf{%
\varepsilon }_{T},  \tag{C.1}
\end{equation}%
with $\mathbf{y}_{T}=(y_{1T},\ldots ,y_{lT})^{\prime }$, $\mathbf{%
\varepsilon }_{T}=(\varepsilon _{1T},\ldots ,\varepsilon _{lT})^{\prime }$,
where the first subscript refers to the season ($s$) and the second one to
the `period' ($T$). Moreover, $\mathbf{\Phi }_{0}$ is an $l\times l$
parameter matrix whose ($i,j$) entry is:%
\begin{equation*}
\left \{ 
\begin{array}{lll}
1 & \text{if} & i=j, \\ 
0 & \text{if} & j>i, \\ 
-\phi _{i-j,i} & if & j<i,%
\end{array}%
\right.
\end{equation*}%
and $\mathbf{\Phi }_{1}$ is an $l\times l$ parameter matrix with ($i,j$)
elements $\phi _{i+l-j,i}$, (see, for example, Lund and Basawa, 2000, and
Franses and Paap, 2005).

As pointed out by Franses (1994), the idea of stacking was introduced by
Gladyshev (1961) and is also considered in e.g., Pagano (1978). Tiao and
Guttman (1980), Osborn (1991) and Franses (1994) used it in the AR setting.
The dynamic system in eq. (C.$1$) can be written in a compact form%
\begin{equation*}
\mathbf{\Phi }(B)\mathbf{y}_{T}\mathbf{=\varepsilon }_{T}\text{ or }\left
\vert \mathbf{\Phi }(B)\right \vert \mathbf{y}_{T}\mathbf{=}adj[\mathbf{\Phi 
}(B)]\mathbf{\varepsilon }_{T}
\end{equation*}%
where $\mathbf{\Phi }(B)=\mathbf{\Phi }_{0}-\mathbf{\Phi }_{1}(B)$.
Stationarity of $\mathbf{y}_{T}$ requires the roots of $\left \vert \mathbf{%
\Phi }(z^{-1})\right \vert =0$ to lie strictly inside the unit circle (see,
among others, Tiao and Guttman, 1980, Osborn, 1991; Franses, 1994, 1996a;
Franses and Paap, 2005; del Barrio Castro and Osborn, 2008).

As an example, consider the PAR($2;4$) model 
\begin{equation*}
y_{t_{s}}=\phi _{1,s}y_{t_{s}-1}+\phi _{2,s}y_{t_{s}-2}+\varepsilon _{t_{s}},
\end{equation*}%
which can be written as%
\begin{equation*}
\mathbf{\Phi }_{0}\mathbf{y}_{T}=\mathbf{\Phi }_{1}\mathbf{y}_{T-1}+\mathbf{%
\varepsilon }_{T},
\end{equation*}%
for which the characteristic equation is%
\begin{equation*}
\left \vert \mathbf{\Phi }_{0}-\mathbf{\Phi }_{1}z\right \vert =\left \vert 
\begin{array}{llll}
1 & 0 & -\phi _{2,1}z & -\phi _{1,1}z \\ 
-\phi _{1,2} & 1 & 0 & -\phi _{2,2}z \\ 
-\phi _{2,3} & -\phi _{1,3} & 1 & 0 \\ 
0 & -\phi _{2,4} & -\phi _{1,4} & 1%
\end{array}%
\right \vert =0.
\end{equation*}%
Hence, when the nonlinear parameter restriction 
\begin{eqnarray*}
&&\left \vert \phi _{2,2}\phi _{1,3}\phi _{1,4}+\phi _{2,2}\phi _{2,4}+\phi
_{2,1}\phi _{1,2}\phi _{1,3}+\phi _{2,1}\phi _{2,3}+\phi _{1,1}\phi
_{1,2}\phi _{1,3}\phi _{1,4}\right. \\
&&\left. +\phi _{1,1}\phi _{1,2}\phi _{2,4}+\phi _{1,1}\phi _{1,4}\phi
_{2,3}-\phi _{2,1}\phi _{2,2}\phi _{2,3}\phi _{2,4}\right \vert <1,
\end{eqnarray*}%
is imposed on the parameters, the VS representation of the PAR($2;4$) model
is stationary (see Franses and Paap, 2005). When $\phi _{2,s}=0$ for all $s$%
, that is we have the PAR($1;4$) model, then the stationarity condition
reduces to: $\left \vert \phi _{1,1}\phi _{1,2}\phi _{1,3}\phi _{1,4}\right
\vert <1$, which is equivalent to our condition $\left \vert \xi
_{t,l}\right \vert <1$, or in other words, that the absolute value of $\left
\vert \mathbf{\Phi }(l)\right \vert $ is less than one.

\newpage

\end{document}